\documentclass{vldb}

\usepackage{tikz}
\usetikzlibrary{trees}
\usetikzlibrary{matrix}

\tikzset{my arrow/.style={blue!60!black,-latex},
  set@com@col/.style={},set@com@col@aryarg/.style={column #1/.style={set@com@col}},
  set@com@row/.style={},set@com@row@aryarg/.style={row #1/.style={set@com@row}},
  set common column/.style 2 args={set@com@col/.style={#2}, set@com@col@aryarg/.list={#1}},
  set common row/.style 2 args={set@com@row/.style={#2}, set@com@row@aryarg/.list={#1}},
}

\usepackage{float}
\usepackage{graphicx}
\usepackage{balance}  
\usepackage{multirow}
\usepackage{booktabs}
\usepackage{hyperref}
\usepackage{nicefrac}

\vldbTitle{Motivo: fast motif counting via succinct color coding and adaptive sampling}
\vldbAuthors{M.\ Bressan, S.\ Leucci, A.\ Panconesi}
\vldbDOI{https://doi.org/TBD}
\vldbVolume{12}
\vldbNumber{xxx}
\vldbYear{2019}

\usepackage{bbm}
\usepackage{amsmath}
\usepackage{amssymb}
\usepackage{algorithm}
\usepackage[noend]{algpseudocode}
\usepackage{relsize}
\usepackage{xspace}

\begin{document}
\title{Motivo: fast motif counting via succinct \\ color coding and adaptive sampling}

\numberofauthors{3}
\author{
\alignauthor Marco Bressan\thanks{\small Supported in part by the ERC Starting Grant DMAP 680153, a Google Focused Research Award, and by the MIUR grant ``Dipartimenti di eccellenza 2018-2022'' of the Dept.\ of Computer Science of Sapienza.}
\\
\affaddr{Dipartimento di Informatica, Sapienza Università di Roma} \\
\email{bressan@di.uniroma1.it} 
\and
\alignauthor Stefano Leucci \\
\affaddr{Department of Algorithms and Complexity, MPI-INF} \\
\email{\makebox[0pt]{stefano.leucci@mpi-inf.mpg.de}}
\and
\alignauthor Alessandro Panconesi\footnotemark[1]\\
\affaddr{Dipartimento di Informatica, Sapienza Università di Roma} \\
\email{ale@di.uniroma1.it}
}

\newtheorem{lemma}{Lemma}
\newtheorem{theorem}{Theorem}
\newtheorem{hypothesis}{Hypothesis}
\newtheorem{corollary}{Corollary}
\newtheorem{proposition}{Proposition}

\newcommand{\dsstyle}[1]{\textsc{\small{#1}}}
\newcommand{\infroadusa}{\dsstyle{Road-US}}
\newcommand{\facebook}{\dsstyle{Facebook}}
\newcommand{\orkut}{\dsstyle{Orkut}}
\newcommand{\yelp}{\dsstyle{Yelp}}
\newcommand{\hollywood}{\dsstyle{Hollywood}}
\newcommand{\lj}{\dsstyle{LiveJournal}}
\newcommand{\wordassoc}{\dsstyle{WordAssoc}}
\newcommand{\twitter}{\dsstyle{Twitter}}
\newcommand{\skitter}{\dsstyle{Skitter}}
\newcommand{\patents}{\dsstyle{Patents}}
\newcommand{\amazon}{\dsstyle{Amazon}}
\newcommand{\fs}{\dsstyle{Friendster}}
\newcommand{\dblp}{\dsstyle{Dblp}}
\newcommand{\roadca}{\dsstyle{Road-CA}}
\newcommand{\roadpa}{\dsstyle{Road-PA}}
\newcommand{\bstan}{\dsstyle{Berk\-Stan}}
\newcommand{\LAW}{\small{LAW}}
\newcommand{\SWS}{\small{MPI-SWS}}
\newcommand{\NDR}{\small{NDR}}
\newcommand{\SNAP}{\small{SNAP}}
\newcommand{\YLP}{\small{YLP}}
\newcommand{\prob}{\operatorname{Pr}}
\newcommand{\bfx}{\mathbf{x}}
\newcommand{\bfc}{\mathbf{c}}
\newcommand{\bfy}{\mathbf{y}}
\newcommand{\E}{\mathbb{E}}
\newcommand{\gk}{\mathcal{G}_k}
\newcommand{\vk}{\mathcal{V}}
\newcommand{\ek}{\mathcal{E}}
\newcommand{\gnk}{\mathcal{G}}
\newcommand{\var}[1]{\operatorname{Var}[#1]}
\newcommand{\cov}{\operatorname{Cov}}
\newcommand{\cerr}{\bar{c}}
\newcommand{\glet}{\ensuremath{x}}
\newcommand{\motivo}{\textsc{motivo}}
\newcommand{\cc}{\textsc{cc}}
\newcommand{\Rplus}{\protect\hspace{-.1em}\protect\raisebox{.35ex}{\smaller{\smaller\textbf{+}}}}
\newcommand{\Cpp}{\mbox{C\Rplus\Rplus}\xspace}
\newcommand{\labelspacing}{6pt}
\hyphenation{graph-let}
\hyphenation{graph-lets}
\newcommand{\sample}{\ensuremath{\operatorname{sample}}}

\newcommand{\tc}{T_C}
\newcommand{\tca}{T'_{C'}}
\newcommand{\tcb}{T''_{C''}}
\newcommand{\ta}{T'}
\newcommand{\tb}{T''}

\maketitle
\begin{abstract}
The randomized technique of color coding is behind state-of-the-art algorithms for estimating graph motif counts.
Those algorithms, however, are not yet capable of scaling well to very large graphs with billions of edges. In this paper we develop novel tools for the ``motif counting via color coding'' framework.
As a result, our new algorithm, \motivo, is able to scale well to larger graphs while at the same time provide more accurate graphlet counts than ever before. This is achieved thanks to two types of improvements.
First, we design new succinct data structures that support fast common color coding operations, and a biased coloring trick that trades accuracy versus running time and memory usage.
These adaptations drastically reduce the time and memory requirements of color coding.
Second, we develop an adaptive graphlet sampling strategy, based on a fractional set cover problem, that breaks the additive approximation barrier of standard sampling.
This strategy gives multiplicative approximations for all graphlets at once, allowing us to count not only the most frequent graphlets but also extremely rare ones.

To give an idea of the improvements, in $40$ minutes \motivo\ counts $7$-nodes motifs on a graph with $65$M nodes and $1.8$B edges; this is $30$ and $500$ times larger than the state of the art, respectively in terms of nodes and edges.
On the accuracy side, in one hour \motivo\ produces accurate counts of $\approx \! 10.000$ distinct $8$-node motifs on graphs where state-of-the-art algorithms fail even to find the second most frequent motif.
Our method requires just a high-end desktop machine.
These results show how color coding can bring motif mining to the realm of truly massive graphs using only ordinary hardware.
\end{abstract}

\section{Introduction}
Graphlets, also called motifs or patterns, are small induced subgraphs of a graph.
Graphlets are often considered the ``building blocks'' of networks~\cite{Jha&2015,interactome,Yaveroglu&2014,Yin&2017}, and their analysis has helped understanding network evolution~\cite{Abdelzaher&2015}, designing better graph classification algorithms~\cite{Yaveroglu&2014}, and developing cutting-edge clustering techniques~\cite{Yin&2017}.

A fundamental problem in graphlet mining and analysis is graphlet counting: estimating as accurately as possible the number of copies of a given graphlet (e.g., a tree, a clique, etc.)\ in a graph.
Graphlet counting has a long and rich history, which began with triangle counting and received intense interest in recent years~\cite{Ahmed&2015,Bhuiyan&2012,Chakaravarthy&2016,Chen&2016,Han&2016,Jha&2015,Pinar&2017,Slota&2013,Wang&2014,Wang&2015,Wang&2016,Zhao&2010}.
Since exact graphlet counting is notoriously hard, one must resort to approximate probabilistic counting to obtain algorithms with an acceptable practical performance.
Approximate counting is indeed often sufficient, for example when performing hypothesis testing (deciding if a graph comes from a certain distribution or not) or estimating the clustering coefficient of a graph (the fraction of triangles among $3$-node graphlets).

The simplest formulation of approximate graphlet counting, which we adopt in this work, is the following.
We are given a simple graph $G$ on $n$ nodes, an integer $k > 2$, and two approximation parameters $\epsilon,\delta \in (0,1)$.
For each  graphlet $H$ on $k$ nodes (the clique, the path, the star etc.), we want a very reliable estimate and accurate estimate of the number of induced copies of $H$ in $G$:
with probability at least $1-\delta$, all estimates should be within a factor $(1 \pm \epsilon)$ of the actual values.
Note that we are talking about induced copies; non-induced copies are easier to count and can be derived from the induced ones.
Our goal is to develop practical algorithms that solve this problem for sizes of $G$ and $H$ that were out of reach before, {\em i.e.}\ graphs with hundreds of millions of edges and graphlets on more than $5$ and $6$ nodes.
Note that the task becomes quickly demanding as $k$ grows; for example, for $k=8$ the number of distinct graphlets is over $10$k, and for $k=10$ over $11.7$M.
Thus, scaling from ``small'' graphlets to ``large'' graphlets likely requires new ideas.

A quick review of existing approaches may help appreciate the state of the art and the main obstacles.
A natural approach to the problem consists in sampling graphlets from $G$, and indeed all known efficient algorithms follow this route.
A popular technique for sampling is to define a random walk over the set of graphlets of $G$, simulate it until it reaches stationarity, and take the last graphlet~\cite{Bhuiyan&2012,Chen&2016,Han&2016,Wang&2014}.
This technique is simple and has a small memory footprint.
However, it cannot estimate graphlet \emph{counts}, but only their frequencies.
Moreover, the random walk may need $\Omega(n^{k-1})$ steps to reach stationarity even if $G$ is fast-mixing~\cite{Bressan&2017,Bressan&2018b}.

An alternative approach that overcomes these limitations was proposed in~\cite{Bressan&2017}.
It extends the color coding technique of Alon et al.~\cite{Alon&1995} by making two key observations.
First, via color coding one can build an abstract ``urn'' which contains a sub-population of all the $k$-trees of $G$ that is very close to the true one.
Second, the problem of sampling $k$-graphlet occurrences can be reduced, with minimal overhead, to sampling $k$-tree occurrences from the urn.
One can thus estimate graphlet counts in two  steps: the \emph{build-up phase}, where one builds the urn from $G$, and the \emph{sampling phase}, where one samples $k$-trees from the urn.
Building the urn requires time $O(a^k m)$ and space $O(a^k n)$ for some $a>0$, where $n$ and $m$ are the number of nodes and edges of $G$, while sampling from the urn takes a variable but typically small amount of time per sample.
The resulting algorithm, dubbed CC in~\cite{Bressan&2017}, outperforms random walk-based approaches and is the current state of the art in approximate motif counting~\cite{Bressan&2017,Bressan&2018b}.

Although CC has extended the outreach of graphlet counting techniques, it cannot effectively cope with graphs with billions of edges and values of $k$ beyond six. This is due to two  main bottlenecks.
First, the time and space taken by the build-up phase are significant and prevent CC from scaling to the values of $G$ and $k$ that we are interested in this paper.
For example, on a machine with 64GB of main memory, the largest graph for which CC runs successfully has $5.4$M nodes for $k=5,6$ and just $2$M nodes for $k=7$.
Second, taking $s$ samples from the abstract urn gives the usual additive $\nicefrac{1}{s}$-approximation, which means we can accurately count only those graphlets whose occurrences are a fraction at least $\nicefrac{1}{s}$ of the total.
Unfortunately, in many graphs most graphlets have a very low relative frequency, and CC is basically useless to count them.

In this work we overcome the limitations of CC by making two main contributions to the ``motif counting via color coding'' framework.
The first contribution is reducing the running time and space usage of the build-up phase.
We do so in three ways.
First, we introduce succinct color coding data structures that can represent colored rooted trees on up to $16$ nodes with just one machine word, and support frequent operations (e.g.\ merging trees) in just a few elementary CPU instructions.
This is key, as colored trees are the main objects manipulated in the build-up phase.
Second, for large graphs we present a simple ``biased coloring'' trick that we use to trade space and time against the accuracy of the urn (the distance of the urn's distribution from the actual tree distribution of $G$), whose loss we quantify via concentration bounds.
Third, we describe a set of architectural and implementation optimizations.
These ingredients make the build-up phase significantly faster and bring us from millions to billions of edges and from $k=5$ to $k=8$.

Our second contribution is for the sampling phase and is of a fundamentally different nature.
To convey the idea, imagine having an urn with 1000 balls of which 990 red, 9 green, and 1 blue.
Sampling from the urn, we will quickly get a good estimate of the fraction of red balls, but we will need many samples to witness even one green or blue ball.
Now imagine that, after having seen those red balls, we could remove from the urn 99\% of all red balls.
We would be left with 10 red balls, 9 green balls, and 1 blue ball.
At this point we could quickly get a good estimate of the fraction of green balls.
We could then ask the urn to delete almost 99\% of the red and green balls, and we could quickly estimate the fraction of blue balls.
What we show here is that the urn built in the build-up phase can be used to perform essentially this ``deletion'' trick, where the object to be removed are treelets.
In this way, roughly speaking, we can first estimate the most frequent graphlet, then delete it from the urn and proceed to the second most frequent graphlet, delete it from the urn and so on.
This means we can in principle obtain a small \emph{relative} error for all graphlets, independently of their relative abundance in $G$, thus breaking the $\Theta(1/\epsilon)$ barrier of standard sampling.
We name this algorithm AGS (adaptive graphlet sampling).
To obtain AGS we actually develop an online greedy algorithm for a fractional set cover problem.
We provide formal guarantees on the accuracy and sampling efficiency of AGS via set cover analysis and martingale concentration bounds.

In order to properly assess the impact of the various optimizations, in this paper we have added them incrementally to CC, which acts as a baseline.  In this way, it is possible to assess in a quantitative way the improvements due to the various components. 

Our final result is an algorithm, \motivo\footnote{The \Cpp\ source code of \motivo\ is publicly available at \hbox{\url{https://bitbucket.org/steven_/motivo}}.}, that scales well beyond the state of the art in terms of input size and simultaneously ensures tighter guarantees.
To give an idea, for $k=7$ \motivo\ manages graphs with tens of millions of nodes and billions of edges, the largest having $65$M nodes and $1.8$B edges.
This is 30 times and 500 times (respectively in terms of $n$ and $m$) what CC can manage.
For $k=8$, our largest graph has $5.4$M nodes and $50$M edges (resp.\ 18 and 55 times CC).
All this is done in $40$ minutes on just a high-end commodity machine.
For accuracy, the most extreme example is the \yelp\ graph, where for $k=8$ all but two graphlets have relative frequency below $10^{-7}$.
With a budget of $1$M samples, CC finds only the \emph{first} graphlet and misses all the others.
\motivo\ instead outputs accurate counts ($\epsilon \le 0.5$) of more than $90\%$ of all graphlets, or $10.000$  in absolute terms.
The least frequent ones of those graphlets have frequency below $10^{-20}$, and CC would need $\sim \! 3\cdot10^3$ years to find them even if it took one billion samples per second.

\subsection{Related work}
\label{sub:rel}
Counting induced subgraphs is a classic problem in computer science.
The exact version is notoriously hard; even detecting a $k$-clique in an $n$-node graph requires time $n^{\Omega(k)}$ under the Exponential Time Hypothesis~\cite{Chen&2006}.
It is not surprising then that practical exact counting algorithms exist only for $k \le 5$.
The fastest such algorithm is currently ESCAPE~\cite{Pinar&2017}, which can take a week on graphs with a few million nodes.
When possible we use it for our ground truth.

For approximate graphlet counting many techniques have been proposed.
For $k \le 5$, one can sample graphlets via path sampling (do a walk on $k$ nodes in $G$ and check the subgraph induced by those nodes)~\cite{Jha&2015,Wang&2015,Wang&2016}.
This technique, however, does not scale to $k>5$.
A popular approach is to sample graphlets via random walks~\cite{Bhuiyan&2012,Wang&2014,Chen&2016,Han&2016}.
The idea is to define two graphlets as adjacent in $G$ if they share $k-1$ nodes.
This implicitly defines a reversible Markov chain over the graphlets of $G$ which can be simulated efficiently.
Once at stationarity, one can take the sample and easily compute an unbiased estimator of the graphlet frequencies.
Unfortunately, these algorithms cannot estimate counts, but only frequencies.
Even then, they may give essentially no guarantee unless one runs the walk for $\Omega(n^{k-1})$ steps, and in practice they are outperformed by CC~\cite{Bressan&2017,Bressan&2018b}.
Another recent approach is that of edge-streaming algorithms based on reservoir sampling~\cite{DeStefani&2017,DeStefani&2017b}, which however are tailored to $k \le 5$.
As of today, the state of the art in terms of $G$ and $k$ is the color-coding based CC algorithm of~\cite{Bressan&2017,Bressan&2018b}.
CC can manage graphs on $\sim \! 5$M nodes for $k=5,6$, on $\sim \!2$M nodes for $k=7$, and on less than $0.5$M nodes for $k=8$, in a matter of minutes or hours.
As said above, CC does not scale to massive graphs and suffers from the ``naive sampling barrier'' that allows only for additive approximations.
Finally, we shall mention the algorithm of~\cite{Jain&2017} that in a few minutes can estimate clique counts with high accuracy on graphs with tens of millions of edges.
We remark that that algorithm works \emph{only} for cliques, while \motivo\ is general purpose and provides counts for \emph{all} graphlets at once.

\smallskip
\textbf{Preliminaries and notation.}
We denote the host graph by $G=(V,E)$, and we let $n=|V|$ and $m=|E|$.
A \emph{graphlet} is a connected graph $H=(V_H,E_H)$.
A \emph{treelet} $T$ is a graphlet that is a tree.
We denote $k=|V_H|$.
We denote by $\mathcal{H}$ the set of all $k$-node graphlets, i.e.\ all non-isomorphic connected graphs on $k$ nodes.
When needed we denote by $H_i$ the $i$-th graphlet of $\mathcal{H}$.
A colored graphlet has a color $c_u \in [k]$ associated to each one of its nodes $u$.
A graphlet is \emph{colorful} if its nodes have pairwise distinct colors.
We denote by $C \subseteq [k]$ a subset of colors.
We denote by $(T,C)$ or $\tc$ a colored treelet whose nodes span the set of colors $C$; we only consider colorful treelets, i.e.\ the case $|T|$=$|C|$.
Often treelets and colored treelets are rooted at a node $r \in T$.

\smallskip
\textbf{Paper organization.}
Section~\ref{sec:cc} reviews color coding and the CC algorithm.
Section~\ref{sec:speed} introduces our data structures and techniques for accelerating color coding.
Section~\ref{sec:ags} describes our adaptive sampling strategy. 

\section{Color coding and CC}
\label{sec:cc}
\label{sub:cc}
The color coding technique was introduced in~\cite{Alon&1995} to probabilistically detect paths and trees in a graph.
The CC algorithm of~\cite{Bressan&2017,Bressan&2018b} is an extension of color coding that enables sampling colorful graphlet occurrences from $G$.
It consists of a \emph{build-up phase} and a \emph{sampling phase}.

\subsection{The build-up phase}
\label{sub:build}
The goal of this phase is to build a \emph{treelet count table} that is the abstract ``urn'' used for sampling.
First, we do a \emph{coloring} of $G$: for each $v \in G$ independently, we draw uniformly at random a color $c_v \in [k]$.
We then look at the treelets copies of $G$ that are colorful.
For each $v$ and every rooted colored treelet $\tc$ on up to $k$ nodes, we want a count $c(\tc,v)$ of the number of copies of $\tc$ in $G$ that are rooted in $v$ (note that we mean \emph{non-induced} copies here).
To this end, for each $v$ we initialize $c(\tc, v)=1$, where $T$ is the trivial treelet on $1$ node and $C = \{c_v\}$.
For a $\tc$ on $h>1$ nodes, the count $c(\tc, v)$  is then computed via dynamic programming, as follows.
First, $T$ has a unique decomposition into two subtrees $\ta$ and $\tb$ rooted respectively at the root $r$ of $T$ and at a child of $r$.
The uniqueness is given by a total order over treelets (see next section).
Now, since $\ta$ and $\tb$ are smaller than $T$, their counts have already been computed for all possible colorings and all possible rootings in $G$.
Then  $c(\tc, v)$ is given by (see~\cite{Bressan&2018b}):
\begin{equation}
\label{eqn:decomp}
c(\tc, v) = \frac{1}{\beta_T} \sum_{u \sim v} \sum_{\substack{C' \subset C \\ |C'|=|\ta|}}\!\!\!\! c(\tca, v) \cdot c(\tcb, u)
\end{equation}
where $\beta_T$ is the number of subtrees of $T$ isomorphic to $\tb$ rooted at a child of $r$.
CC employs~\eqref{eqn:decomp} in the opposite way: it iterates over all pairs of counts $c(\tca, v)$ and $c(\tcb, u)$ for all $u \sim v$, and if $\tca, \tcb$ can be merged in a colorful treelet $\tc$, then it adds $c(\tca, v) \cdot c(\tcb, u)$ to the count $c(\tc, v)$.
This requires to perform a check-and-merge operation for each count pair, which is quite expensive (see below).

A simple analysis gives the following complexity bounds:
\begin{theorem}(\cite{Bressan&2018b}, Theorem 5.1).
The build-up phase takes time $O(a^k m)$ and space $O(a^k n)$, for some constant $a>0$. 
\end{theorem}
A major bottleneck is caused by the quick growth of the dynamic programming table: already for $k=6$ and $n=5$M, CC  takes $45$GB of main memory~\cite{Bressan&2018b}.

\subsection{The sampling phase}
\label{sub:sampling}
The goal of this phase is to sample \emph{colorful} graphlet copies u.a.r.\ from $G$, using the treelet count table from the build-up phase.
The key observation (\cite{Bressan&2017,Bressan&2018b}) is that we only need to sample colorful non-induced \emph{treelet} copies; by taking the corresponding induced subgraph in $G$, we then obtain our induced graphlet copies.
Colorful treelets are sampled via a multi-stage sampling, as follows.
First, draw a  node $v \in G$ with probability proportional to $\eta_v = \sum_{\tc}c(\tc,v)$.
Second, draw a colored treelet $\tc$ with probability proportional to $c(\tc, v)/\eta_v$.
We want to sample a copy of $\tc$ rooted at $v$.
To this end we decompose $\tc$ into $\tca$ and $\tcb$, with $\tca$ rooted at the root $r$ of $T$ and $\tcb$ at a child of $r$ (see above).
We then recursively sample a copy of $\tca$ rooted at $v$, and a copy of $\tcb$ rooted at node $u \sim v$, where $u$ is chosen with probability $c(\tc,u)/\sum_{z \sim v}c(\tc,z)$.
Note that computing this probability requires listing all neighbors $z$ of $v$, which takes time proportional to $d_v$.
Finally, we combine $\tca$ and $\tcb$ into a copy of $\tc$.
One can see that this gives a colorful copy of $T$ drawn uniformly at random from $G$.

Consider then a given $k$-graphlet $H_i$ (e.g.\ the clique), and let $c_i$ be the number of colorful copies of $H_i$ in $G$.
We can  estimate $c_i$ as follows.
Let $\chi_i$ be the indicator random variable of the event that a graphlet sample $\glet$ is an occurrence of $H_i$.
It is easy to see that $\E[\chi_i] = c_i \,\sigma_i / t$, where $\sigma_i$ is the number of spanning trees in $H_i$ and  $t$ is the total number of colorful $k$-treelets of $G$.
Both $t$ and $\sigma_i$ can be computed quickly, by summing over the treelet count table and via Kirchhoff's theorem (see below).
We thus let $\hat{c}_i = t\, \sigma_i^{-1} \chi_i$, and $\E[\hat{c}_i] = c_i$.
By standard concentration bounds we can then estimate $c_i$ by repeated sampling.
Note that the expected number of samples to find a copy of $H_i$ grows as $1 / c_i$.
This is the additive error barrier of CC's sampling.

\textbf{Estimators and errors.}
Finally, let us see how to estimate the number of total (i.e.\ \emph{uncolored}) copies $g_i$ of $H_i$ in $G$, which is our final goal.
First, note that the probability that a fixed subset of $k$ nodes in $G$ becomes colorful is $p_k = k! / k^k$.
Therefore, if $G$ contains $g_i$ copies of $H_i$, and $c_i$ is the number those copies that become colorful, then by linearity of expectation $\E[c_i] = p_k g_i$ (seeing $c_i$ as a random variable).
Hence, $\hat{g}_i = c_i/p_k$ is an unbiased estimator for $g_i$.
This is, indeed, the count estimate returned by CC and by \motivo.

For what concerns accuracy, the error given by $\hat{g}_i$ can be formally bounded via concentration bounds.
An additive error bound is given by Theorem 5.3 of~\cite{Bressan&2018b}, which we slightly rephrase.
Let $g = \sum_i g_i$ be the total number of induced $k$-graphlet copies in $G$.
Then:
\begin{theorem}[\cite{Bressan&2018b}, Theorem 5.3]
\label{thm:color_conc}
For all $\epsilon>0$,
\[
\prob\!\Big[\big|\hat{g}_i - g_i\big| > \frac{2 \epsilon g }{1-\epsilon}\Big] = \exp(-\Omega( \epsilon^2 g^{1/k} ))
\]
\end{theorem}
Since we aim at multiplicative errors, we prove a multiplicative bound, which is also tighter than Theorem~\ref{thm:color_conc} if the maximum degree $\Delta$ of $G$ is small.
We prove (see Appendix~\ref{apx:conc_dep}):
\begin{theorem}
\label{THM:CONC_DEP}
For all $\epsilon > 0$,
\begin{align}
\prob\!\Big[\big|\hat{g}_i - g_i\big| > \epsilon \,g_i \Big] < 2 \exp\!\Big(\!-\frac{2\epsilon^2 }{(k-1)!} \frac{p_k \, g_i }{\Delta^{k-2}}\Big)
\end{align}
\end{theorem}
In practice, $\hat{g}_i$ appears always concentrated.
In other words the coloring does not introduce a significant distortion.
Moreover, if one averages over $\gamma$ independent colorings, the probabilities in the bounds decrease exponentially with $\gamma$.

\section{Speeding up color coding}
\label{sec:speed}
We detail step-by-step the data structures and optimizations that are at the heart of \motivo's efficiency.
As a baseline for our comparisons, we ported CC in \Cpp\ (CC is originally written in Java), using the sparse hash tables from the sparsehash library\footnote{\url{https://github.com/sparsehash/sparsehash}}.
We then incrementally added/replaced its components, measuring their impact as we move from the porting of CC to \motivo.\footnote{The baseline \Cpp\ porting is actually \emph{slower} than CC, sometimes by an order of magnitude. We suspect this is due to the hash tables (CC uses \url{http://fastutil.di.unimi.it/}).}

\subsection{Succinct data structures}
\label{sub:ds}
The main objects manipulated by CC and \motivo\ are rooted colored treelets and their associated counts, which are stored in the treelet count table.
We first describe their implementation in CC, then introduce the one of \motivo.

\textbf{The internals of CC.}
In CC, each $\tc$ has a unique \emph{representative instance}, that is a classic pointer-based tree data structure equipped with a structure storing the colors.
The pointer to this instance acts as unique identifier for $\tc$.
The treelet count table of CC is then implemented as follows: for each $v \in G$, a hash table maps the pointer of each $\tc$ to the count $c(\tc,v)$, provided $c(\tc,v)>0$.
Thus, each entry uses 128 bits -- 64 for the pointer and 64 for the count -- plus the overhead of the hash table.
For computing $c(\tc,v)$, CC processes every neighbor $u \sim v$ as follows (see also Section~\ref{sub:build}).
For every pair of counts $c(T'_{C'}, v)$ and $c(T''_{C''}, u)$ in the hash tables of $v$ and $u$, check that $C' \cap C'' = \emptyset$, and that $T''_{C''}$ comes before the smallest subtree of $T'_{C'}$ in the total order of the treelets (see below).
If these conditions hold, then $T'_{C'}$ and $T''_{C''}$ can be merged into a treelet $\tc$ whose unique decomposition yields precisely $T'_{C'}$ and $T''_{C''}$.
Then, the value of $c(\tc,v)$ in the hash table of $v$ is incremented by $c(T'_{C'},v)	 \cdot c(T''_{C''},u)$.
The expensive part is the check-and-merge operation, which CC does with a recursive algorithm on the treelet representative instances.
This has a huge impact, since on a graph with a billion edges the check-and-merge is easily performed trillions of times.

\textbf{Motivo's treelets.}
Let us now describe \motivo's data structures, starting with an uncolored treelet $T$ rooted at $r \in T$.
We encode $T$ with the binary string $s_{T}$ defined as follows. 
Perform a DFS traversal of $T$ starting from $r$. 
Then the $i$-th bit of $s_{T}$ is $1$ (resp.\ $0$) if the $i$-th edge is traversed moving away from (resp.\ towards) $r$.
For all $k \le 16$, this encoding takes at most $30$ bits, which fits nicely in a $4$-byte integer type (padded with $0$s).
The lexicographic ordering over the $s_{T}$'s gives a total ordering over the $T$'s that is exactly the one used by CC.
This ordering is also a tie-breaking rule for the DFS traversal: the children of a node are visited in the order given by their rooted subtrees.
This implies that every $T$ has a well-defined unique encoding $s_{T}$.
Moreover, merging $T'$ and $T''$ into $T$ requires just concatenating $1,s_{T''},s_{T'}$ in this order.
This makes check-and-merge operations extremely fast (see below).

This succinct encoding supports the following operations:
\begin{itemize}
\item \texttt{getsize()}: return the number of vertices in $T$. This is one plus the Hamming weight of $s_{T}$, which can be computed in a single machine instruction (e.g., \texttt{POPCNT} from the SSE4 instruction set).
\item \texttt{merge($\ta$, $\tb$)}: merge two treelets $\ta$, $\tb$ by appending $\tb$ as a child of the root of $\ta$. This requires just to concatenate $1,s_{\tb},s_{\ta}$ in this order.
\item \texttt{decomp($T$)}: decompose $T$ into $\ta$ and $\tb$. This is the inverse of \texttt{merge} and is done by suitably splitting $s_{T}$.
\item \texttt{sub($T$)}: compute the value  $\beta_T$ of \eqref{eqn:decomp}, i.e.\ the number of subtrees of $T$ that (i) are isomorphic to the treelet $\tb$ of the decomposition of $T$, and (ii) are rooted at some child of the root. 
This is done via bitwise \texttt{shift} and \texttt{and} operations on $s_{T}$.
\end{itemize}

A colored rooted treelet $\tc$ is encoded as the concatenation $s_{\tc}$ of $s_{T}$ and of the characteristic vector $s_C$ of $C$.\footnote{Given an universe $U$, the characteristic vector $\langle x_1, x_2, \dots \rangle$ of a subset $S \subseteq U$ contains one bit $x_i$ for each element $i \in U$, which is $1$ if $i \in S$ and $0$ otherwise.}
For all $k \le 16$, $s_{\tc}$ fits in 46 bits.
Set-theoretical operations on $C$ become bitwise operations over $s_C$ (\texttt{or} for union, \texttt{and} for intersection).
Finally, the lexicographical order of the $s_{\tc}$'s induce a total order over the $\tc$'s, which we use in the count table (see below).
An example of a colored rooted treelet and its encoding is given in Figure~\ref{fig:treelet1} (each node  labelled with its color).
\begin{figure}[h]
\centering
\scalebox{0.9}{

\begin{tikzpicture}[level distance=0.6cm,
sibling distance=1cm, minimum size=11pt,inner sep=1.5,outer sep=0,
level 1/.style={sibling distance=0.9cm, minimum size=12pt,outer sep=0},
level 2/.style={sibling distance=0.9cm, minimum size=12pt,outer sep=0}]

\tikzstyle{every node}=[circle,draw,minimum size=5pt]

\node (Root)  {3}
    child {
    node {1} 
    child { node {2} }
    child { node {7} }
}
child {
    node {5}
};
\end{tikzpicture}
\hspace*{5pt}
\begin{tikzpicture}
\matrix[matrix of nodes,row sep=0mm,set common column={1,...,16}{nodes={rectangle,draw,minimum width=1em,inner sep=2.8pt}}] (O)
{
1 & 1 & 0 & 1 & 0 & 0 & 1 & 0 &    1 & 0 & 1 & 0 & 1 & 1 & 1 & 0\\
};
\draw[<->](-2.9,-0.4) -- node[below] {$s_{T}$} (-0.1,-0.4);
\draw[<->](0.1,-0.4) -- node[below] {$s_C$} (2.9,-0.4);
\end{tikzpicture}}
\caption{A colored rooted treelet and its encoding, shown for simplicity on just $8+8=16$ bits.
}
\label{fig:treelet1}
\end{figure}
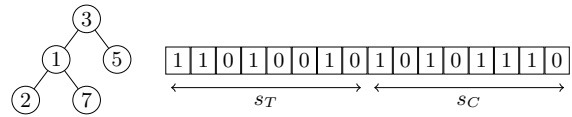

\textbf{Impact.}
The impact of succinct treelets is depicted in Figure~\ref{fig:succint1}, showing the time spent in check-and-merge operations in the build-up phase (single-threaded).
The speedup varies, but is close to $2\times$ on average.
\begin{figure}[h!]
\centering
\includegraphics[scale=.8]{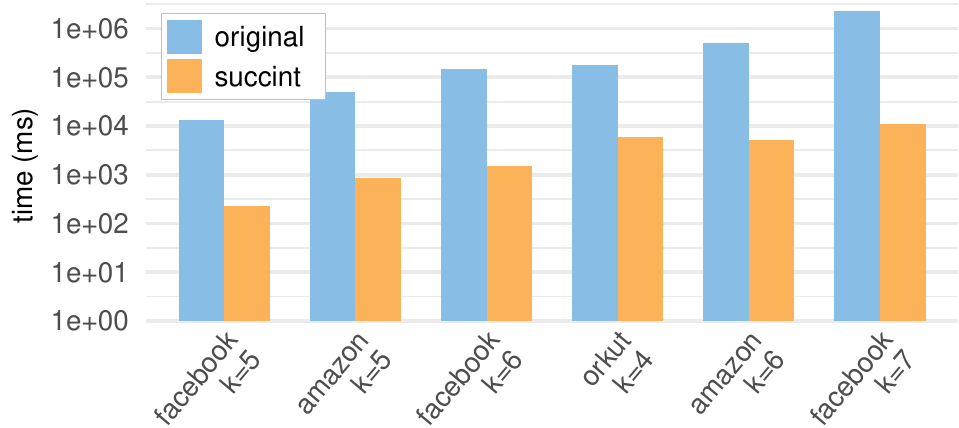}
\caption{impact of succinct treelets: time spent in check-and-merge operations
(logarithmic scale).}
\label{fig:succint1}
\end{figure}

\textbf{Motivo's count table.}
In CC, treelet counts are stored in $n$ hash tables, one for each node $v \in G$.
In each table, the pair $(\tc, c(\tc, v))$ is stored using the pointer to the representative instance of $\tc$ as key.
This imposes the overhead of dereferencing a pointer before each check-and-merge operation to retrieve the actual structure of $\tc$.
Instead of using a hash table, \motivo\ stores the key-value pairs $(s_{\tc}, c(\tc, v))$ in a set of arrays, one for each $v \in G$ and each treelet size $h \in [k]$, sorted by the lexicographical order of the keys $s_{\tc}$.
This makes iterating over the counts extremely fast, and eliminates the need for dereferencing, since each key $s_{\tc}$ is itself an explicit representation of $\tc$.
The result is a large speedup in the build-up phase (see the experiments below).
The price to pay is that searching for a given $\tc$ in the count table requires a binary search.
However, this still takes only $O(k)$, since the whole record has length $O(6^k)$.\footnote{By Cayley's formula: there are $O(3^k k^{-3/2})$ rooted treelets on $k$ vertices~\cite{otter1948number}, and $2^k$ subsets of $k$ colors.}
Note that \motivo\ uses 128-bit counts\footnote{Tests on our machine show that summing 500k unsigned integers is 1.5$\times$ slower with 128-bit than with 64-bit integers.}, whereas CC uses 64-bit counts which often cause overflows (consider that just the number of $6$-stars centered in a node of degree $2^{16}$ is $\approx 2^{80}$).
This increases by 64 bits the space per pair compared to CC; however, \motivo\ saves 16 bits per pair by packing $s_{\tc}$ into 48 bits, using a total of 176 bits per pair.
Finally, in place of $c(\tc, v)$, \motivo\ actually stores the cumulative count $\eta(\tc,v) = \sum_{T'_{C'} \le \tc} c(T'_{C'}, v)$.
In this way each $c(\tc, v)$ can be recovered with negligible overhead, and the total count for a single node $v$ (needed for sampling) is just at the end of the record.

\motivo's count table supports the following operations:
\begin{itemize}
\item \texttt{occ($v$)}: return the total number of colorful treelet occurrences rooted at $v$. Running time $O(1)$.
\item \texttt{occ($\tc, v$)}: return the number of occurrences of $\tc$ rooted at $v$. Running time $O(k)$ via binary search. 
\item \texttt{iter($T,v$)}: get an iterator to the counts of an uncolored treelet $T$ rooted at $v$. Running time $O(k)$, plus $O(1)$ per accessed treelet.
\item \texttt{iter($\tc,v$)}: get an iterator to the counts of a colored treelet $\tc$ rooted at $v$. Running time $O(k)$, plus $O(1)$ per accessed treelet.
\item \texttt{sample($v$)}: returns a random colored treelet $\tc$ with probability proportional to $c(\tc, v)/ \eta_v$. This is used in the sampling phase. Running time $O(k)$: first we get $\eta_v$ in $O(1)$ (see above), then in $O(k)$ we draw $R$ u.a.r.\ from $\{1, \dots, \eta_v \}$, we search the first pair $(\tc, \eta)$ with $\eta \ge R$, and we return $\tc$.
\end{itemize}

\textbf{Greedy flushing.}
The compact treelet count table allows us to match the memory used by CC after porting it in \Cpp\ (see above), but with large gains in computing time and with 128-bit counts support.
To further reduce memory usage, we use an greedy flushing strategy.
Suppose we are currently building the table for treelets of size $h$.
While being built, the record of $v$ is actually stored in a hash table, which allows for efficient insertions.
However, immediately after completion it is stored on disk in the compact form described above, but still unsorted.
The hash table is then emptied and memory released.
When all records have been stored on disk, a second I/O pass sorts them by key.
At the end, the treelet count table is stored on disk without having entirely resided in memory.
In practice, the sorting takes less than 10\% of the total time in all our runs.

\vspace*{1em}
\textbf{Impact.}
Figure~\ref{fig:succint2} compares the \Cpp\ porting of CC before and after adopting succinct treelets, compact count table, and greedy flushing.
The memory usage is given by the maximum resident set size via the Unix command \texttt{time}.
It should be noted that, in our measurements, CC spends $\approx 50\%$ of its running time in check-and-merge operations.
This means that succinct treelets account for roughly a half of the reduction in running time; the rest is brought by the compact count table and greedy flushing.
\begin{figure}[h!]
\centering
\hspace*{-5pt}\includegraphics[scale=.8]{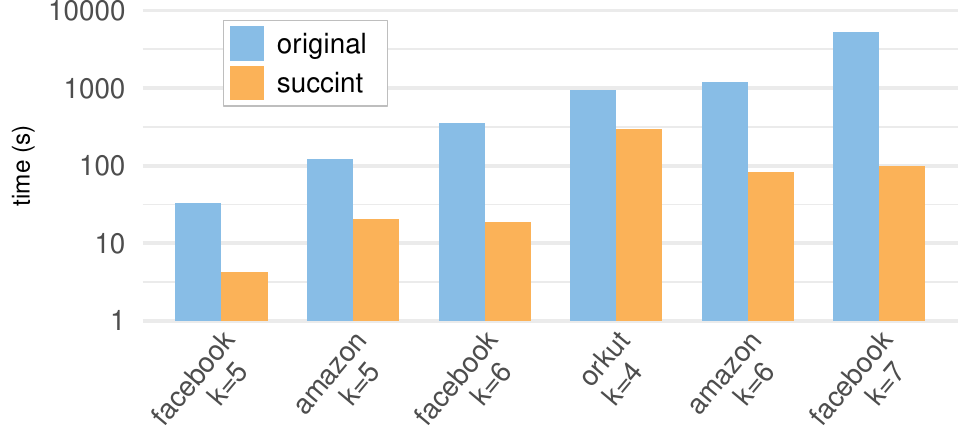}
\\[10pt]
\includegraphics[scale=.79]{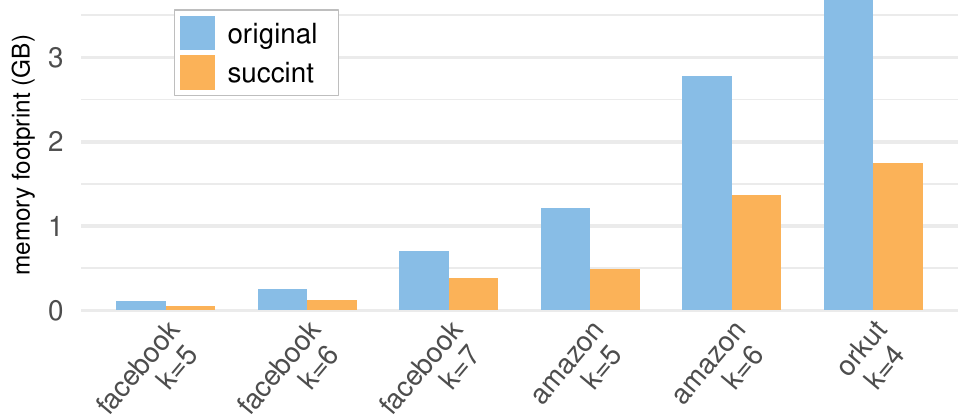}
\caption{impact of succinct treelets, compact count table, and greedy flushing, on the build-up phase.}
\label{fig:succint2}
\end{figure}

\subsection{Other optimizations}
\textbf{0-rooting.}
Consider a colorful treelet copy in $G$ that is formed by the nodes $v_1,\ldots,v_h$.
In the count table, this treelet is counted in each one of the $h$ records of $v_1,\ldots,v_h$, since it is effectively a colorful treelet rooted in each one of those nodes.
Therefore, the treelet is counted $h$ times.
This is inevitable for $h < k$, since excluding some rooting would invalidate the dynamic programming (Equation~\ref{eqn:decomp}).
However, for $h=k$ we can store only one rooting and the sampling works just as fine.
Thus, for $h=k$ we count only the $k$-treelets rooted at their node of color $0$.
This cuts the running time by $30\%-40\%$, while reducing by a factor of $k$ the size of the $k$-treelets records, and by $\approx 10\%$ the total space usage of \motivo.
Figure~\ref{fig:0rooting} depicts the impact of adding $0$-rooting on top of the previous optimizations.
\begin{figure}[h!]
\includegraphics[scale=.78]{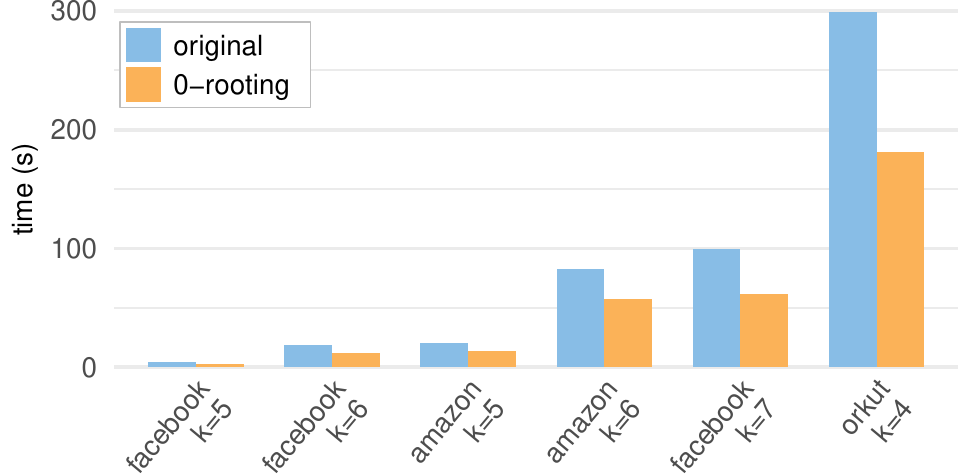}
\caption{impact of $0$-rooting.}
\label{fig:0rooting}
\end{figure}

\textbf{Neighbor buffering.}
Our final optimization concerns sampling.
In most graphs, \motivo\ natively achieves sampling rates of $10$k samples per second or higher.
But on some graphs, such as \bstan\ or \orkut, we get only $100$ or $1000$ samples per second.
The reason is the following.
Those graphs contain a node $v$ with a degree $\Delta$ much larger than any other node.
Inevitably then, a large fraction of the treelets of $G$ are rooted in $v$.
This has two combined effects on the sampling phase (see Subsection~\ref{sub:sampling}).
First, $v$ will be frequently chosen as root.
Second, upon choosing $v$ will spend time $\Theta(\Delta)$ to sweep over its neighbors.
The net effect is that the time to take one sample grows \emph{superlinearly} with $\Delta$, reducing the sampling rate dramatically.
To compensate, we adopt a buffered sampling strategy.
If $d_v \ge 10^4$, then \motivo\ samples $100$ neighbors of $v$ instead of just one, keeping the remaining $99$ cached for future requests.
Sampling $100$ neighbors is as expensive as sampling just one, i.e.\ it takes only a single sweep.
In this way, for large-degree nodes we sweep only $1\%$ of the times.
The impact is depicted in Figure~\ref{fig:buffered}: sampling rates increase by $\approx 20 \times$ on \orkut\ and by $\approx 40\times$ on \bstan.
\begin{figure}[h]
\includegraphics[scale=.78]{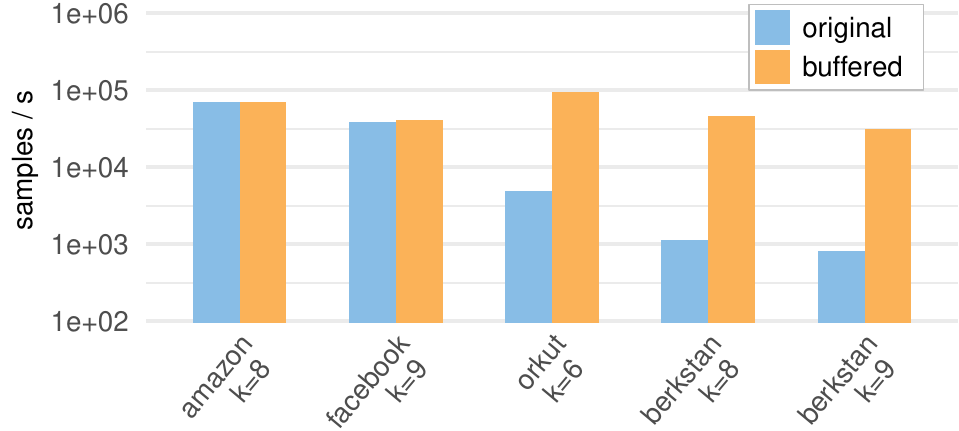}
\caption{impact of neighbor buffering.}
\label{fig:buffered}
\end{figure}

\subsection{Implementation details}
We describe some other implementation details of \motivo\ that, although not necessarily being ``optimizations'', are necessary for completeness and reproducibility.
Whenever possible, we report their impact.

\textbf{Input graph.}
The graph $G$ is stored using the adjacency list representation.
Each list is a sorted static array of the vertex's neighbors; arrays of consecutive vertices are contiguous in memory.
This allows for fast iterations over the set of outgoing edges of a vertex, and for  $O(\log n)$-time edge-membership queries\footnote{This is actually $O(\log \delta(u))$ where $(u,v)$ is the edge being tested, and  $\delta(u)$ is the out-degree of $u$ in $G$. In practice  it is often the case that $\delta(u) \ll n$. }, that we need in the sampling phase to obtain the induced graphlet from the sampled treelet.

\textbf{Multi-threading.}
Similarly to CC, \motivo\ makes heavy use of thread-level parallelism in both the build-up and sampling phases.
For the build-up phase, for any given $v$ the counts $c(\cdot, v)$ can be computed independently from each other, which we do using a pool of threads.
As long as the number of remaining vertices is sufficiently large,
each thread is assigned a (yet unprocessed) vertex $v$ and will compute all the counts $c(\tc, v)$ for all pairs $\tc$.
While this requires minimal synchronization, when the number of unprocessed vertices decreases below the amount of available threads, the above strategy is no longer advantageous as it would cause some of the threads to become idle. 
This can increase the time needed by the build-up phase if $G$ exhibits skewed degree and/or treelet distributions.
To overcome this problem, the last remaining vertices are handled differently: we allow multiple threads to concurrently compute different  summands of the outermost sum of \eqref{eqn:decomp} for the same vertex $v$, i.e., those corresponding to the edges $(v,u) \in E$.
Once all the incident edges of $v$ have been processed, the partial sums are then combined together to obtain all the counts $c(\cdot, v)$.
This reduces the running time by a few percentage points.
For the sampling phase, samples are by definition independent and are taken by different threads.

\textbf{Memory-mapped reads.}
In the build-up phase, to compute the $h$-treelets count table we must access the $j$-treelet count tables for all $j < h$.
For large instances, loading all those tables simultaneously in memory is infeasible.
One option would be to carefully orchestrate I/O and computation, hoping to guarantee a small number of load/store operations on disk.
We adopt a simpler solution: memory-mapped I/O.
This delegates the I/O to the operating system in a manner that is transparent to \motivo, which sees all tables as if they resided in main memory.
When enough memory is available this solution gives ideally no overhead.
Otherwise, the operating system will reclaim memory by unloading part of the tables, and future requests to those parts will incur a page fault and prompt a reload from the disk.
The overhead of this approach can be indeed measured via the number of page faults.
This reveals that the total I/O volume due to page faults is less than 100MB, except for $k=8$ on \lj\ (34GB) and \yelp\ (8GB) and for $k=6$ on \fs\ (15GB).
However, in those cases additional I/O is inevitable, as the total size of the tables (respectively 99GB, 90GB, and 61GB) is close to or even larger than the total memory available (60GB).

\textbf{Alias method sampling.}
Recall that, to sample a colorful graphlet from $G$, we first sample a node $v$ with probability proportional to the number of colorful $k$-treelets rooted at $v$ (Subsection~\ref{sub:sampling}). 
We do this in time $O(1)$ by using the alias method~\cite{Vose91}, which requires building an auxiliary lookup table in time and space linear in the support of the distribution.
In our case this means time and space $O(n)$; the table is built during the second stage of the build-up process.
In practice, building the table takes negligible amounts of time (a fraction of a second out of several minutes).

\textbf{Graphlets.}
In \motivo, each graphlet $H$ is encoded as an adjacency matrix packed in a \hbox{$128$-bit} integer.
Since a graphlet is a simple graph, the $k \times k$ adjacency matrix is symmetric with diagonal $0$ and can be packed in a $(k-1) \times \frac{k}{2}$ matrix if $k$ is even and in
a $k \times \frac{k-1}{2}$ matrix if $k$ is odd (see e.g.\ \cite{BaroudiSL17}). 
The resulting triangular matrix can then be reshaped into a $1 \times \frac{k^2 - k}{2}$ vector, which fits into $120$ bits for all $k \le 16$.
In fact, one can easily compute a bijection between the pair of vertices of the graphlet and the indices $\{1, \dots, 120\}$.
Before encoding a graphlet, \motivo\ replaces it with a canonical representative from its isomorphism class, computed using the Nauty library~\cite{McKay201494}.

\textbf{Spanning trees.}
By default, \motivo\ computes the number of spanning trees $\sigma_i$ of $H_i$ in time $O(k^3)$ via Kirchhoff's matrix-tree theorem which relates $\sigma_i$ to the determinant of a $(k-1) \times (k-1)$ submatrix of the laplacian matrix of $H_i$.
To compute the number $\sigma_{ij}$ of occurrences of $T_i$ in $H_j$ (needed for our sampling algorithm AGS, see Section~\ref{sec:ags}), we use an in-memory implementation of the build-up phase.
The time taken is negligible for $k<7$, but is significant for $k\ge 7$.
For this reason, \motivo\ caches the $\sigma_{ij}$ and stores them to disk for later reuse.
In some cases (e.g.\ $k=8$ on \facebook) this accelerates sampling by an order of magnitude.
$T$; when a new $T$ is chosen, the alias sampler must be rebuilt from scratch.

\subsection{Biased coloring}
\label{sub:biased}
Finally, we describe an optimization, that we call ``biased coloring'', that can be used to manage graphs that would otherwise be too large.
Suppose for simplicity that, for each treelet $T$ on $j$ nodes, each $v \in G$ appears in a relatively small number of copies of $T$, say $k^j/j!$.
Then, given a set $C$ of $j$ colors, a copy of $T$ is colored with $C$ with probability $j! / k^j$.
This implies that we will have an expected $\Theta(1)$ copies of $T$ colored with $C$ containing $v$, in which case the total table size (and the total running time) will approach the worst-case space bounds.

Suppose now we bias the distribution of colors.
In particular, we give probability $\lambda \ll \frac{1}{k}$ to each color in $\{1,\ldots,k-1\}$.
The probability that a given $j$-treelet copy is colored with $C$ is then:
\begin{align}
p_{k,j}(C) = \left\{
\begin{array}{rr}
j! \lambda^j &\; \text{if} \; k \notin C \\
\sim j! \lambda^{j-1} &\; \text{if} \; k \in C
\end{array}
\right.
\end{align}
If $\lambda$ is sufficiently small, then, for most $T$ we will have a zero count at $v$; and most nonzero counts will be for a restricted set of colorings -- those containing $k$.
This reduces the number of pairs stored in the treelet count table, and consequently the running time of the algorithm.
The price to pay is a loss in accuracy, since a lower $p_k$ increases the variance of the number $c_i$ of colorful copies of $H_i$.
However, if $n$ is large enough and most nodes $v$ belong to even a small number of copies of $H_i$, then the \emph{total} number of copies $g_i$ of $H_i$ is large enough to ensure concentration.
In particular, by Theorem~\ref{THM:CONC_DEP} the accuracy loss remains negligible as long as $\lambda^{k-1} n / \Delta^{k-2}$ is large (ignoring factors depending only on $k$).
We can thus trade a $\Theta(1)$ factor in the exponent of the bound for a $\Theta(1)$ factor in both time and space, especially on large graphs where saving resources is precious.

For the choice of $\lambda$, we note one can find a good value as follows.
Start with $\lambda = 1/b(k-1)n$ for some appropriate $b>1$.
By Markov's inequality, with probability $1-\frac{1}{b}$ all $v \in G$ have the same color and thus the table count is empty for all $j$.
Grow $\lambda$ progressively until a small but non-negligible fraction of counts are positive.
Then by Theorem~\ref{THM:CONC_DEP} we have achieved concentration, and we can safely proceed to the sampling phase.

\textbf{Impact.}
With $\lambda = 0.001$, the build-up time on \fs\ ($65$M nodes, $1.8$B edges) shrinks from $17$ to $10$ minutes ($1.7\times$) for $k=5$, and from $1.5$ hours to $13$ minutes ($7\times$) for $k=6$.
In both cases, the main memory usage and the disk space usage decrease by at least $2\times$.
The relative graphlet count error increases correspondingly, as shown in Figure~\ref{fig:biased} (see Section~\ref{sec:exp} for the error definition).
For $k=7$, the build takes $20$ minutes -- in this case we have no comparison term, as without biased coloring \motivo\ did not terminate a run within 2 hours.
Note that \fs\ has 30 (500) times the nodes (edges) of the largest graph managed by CC for the same values of $k$~\cite{Bressan&2018b}.
Note that in our experiments (Section~\ref{sec:exp}) biased coloring is disabled since mostly unnecessary.

\begin{figure}[h]
\includegraphics[scale=.71]{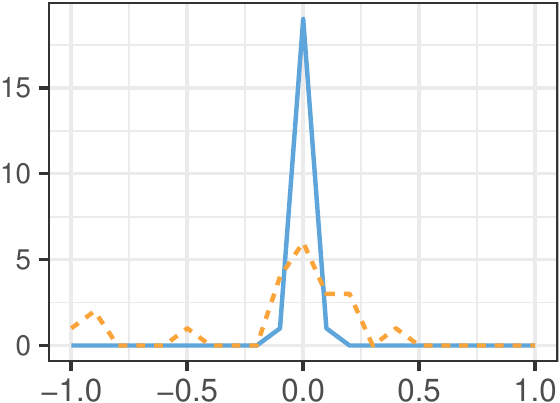}
\hfill
\includegraphics[scale=.71]{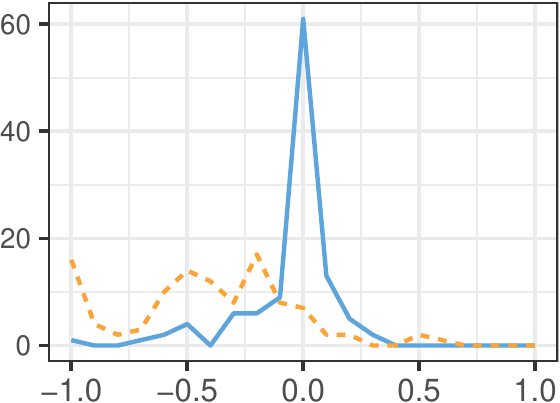}
\caption{Graphlet count error distribution of uniform and biased coloring (dashed), for k=5 and k=6.}
\label{fig:biased}
\end{figure}

\section{Adaptive graphlet sampling}
\label{sec:ags}
This section describes AGS, our adaptive graphlet sampling algorithm for color coding.
Recall that the main idea of CC is to build a sort of ``urn'' supporting a primitive sample() that returns a colorful $k$-treelet occurrence u.a.r.\ from $G$.
The first step of AGS is to ``refine'' this interface with one urn for each possible $k$-treelet shape $T$.
More precisely, for every $k$-treelet shape $T$ our urn should support the following primitive:
\begin{itemize}
\item \sample($T$): return a colorful copy of $T$ u.a.r.\ from $G$
\end{itemize}
We can implement \sample($T$) as explained below.
With \sample($T$) one can selectively sample treelets of different shapes, and this can be used to virtually ``delete'' undesired graphlets from the urn.
Let us try to convey the idea with a simple example.
Imagine $G$ contains just two types of colorful graphlets, $H_1$ and $H_2$, of which $H_2$ represents a tiny fraction $p$ (say 0.01\%).
Using our original primitive, \sample(), we will need $\Theta(1/p)$ calls before finding $H_2$.
Instead, we start using \sample($T_1$), until we estimate accurately $H_1$.
At this point we switch to \sample($T_2$), which completely ignores $H_1$ (since it is not spanned $T_2$), until we estimate accurately $H_2$ as well.
In this way we can estimate accurately both graphlets with essentially $O(1)$ samples.
Clearly, in general we have more than just two graphlets, and distinct graphlets may have the same spanning trees.
Still, our approach gives an adaptive sampling strategy (AGS) that performs surprisingly better than naive sampling in the presence of rare graphlets.
AGS yields multiplicative guarantees on all graphlets, while taking only $O(k^2)$ times the minimum number of samples any algorithm  must take (see below).

We can now turn to describe AGS in more detail.
Initially, we choose the $k$-treelet $T$ with the largest number of colorful occurrences.
Recall from Section~\ref{sub:cc} that for every $v \in G$ we know \texttt{occ($T,v$)}, the number of colorful copies of $T$ rooted at $v$.
Then it is not hard to see that, after some preprocessing, one can restrict the sampling process described in Subsection~\ref{sub:sampling} to the occurrences of $T$, thus drawing u.a.r.\ from the set of colorful copies of $T$.
This gives our primitive \sample($T$).
We then start invoking \sample($T$) until, eventually, some graphlet $H_i$ spanned by $T$ appears enough times, say $\Theta(\frac{1}{\epsilon^2}\ln(\frac{1}{\delta}))$.
We then say $H_i$ is \emph{covered}.
Now, since we do not need more samples of $H_i$, we would like to continue with \sample($T'$) for some $T'$ that does \emph{not} span $H_i$, as if we were asking to ``delete'' $H_i$ from the urn.
More precisely, we seek $T'$ that minimizes the probability that by calling \sample($T'$) we observe $H_i$.

The crux of AGS is that we can find $T'$ as follows.
First, we estimate the number $g_i$ of colorful copies of $H_i$ in $G$, which we can do since we have enough samples of $H_i$.
Then, for each $k$-treelet $T_j$ we estimate the number of copies of $T_j$ that span a copy of $H_i$ in $G$ as $g_i \sigma_{ij}$, where $\sigma_{ij}$ is the number of spanning trees of $H_i$ isomorphic to $T_j$.
We then divide this estimate by the number $r_j$ of colorful copies of $T_j$ in $G$, obtained summing \texttt{occ($T_j,v$)} over all $v \in G$.
The result is an estimate of the probability that \sample($T_j$) spans a copy of $H_i$, and we choose the treelet $T_{j^*}$ that minimizes this probability.
More in general, we need the probability that \sample($T_j$) spans a copy of \emph{some} graphlet among the ones covered so far, and to estimate $g_i$ we must take into account that we have used different treelets along the sampling.

The pseudocode of AGS is listed below. 
A graphlet is marked as covered when it has appeared in at least $\bar{c}$ samples.
For a  union bound over all $k$-graphlets one would set $\cerr = O(\frac{1}{\epsilon^2}\ln(\frac{s}{\delta}))$ where $s=s_k$ is the number of distinct $k$-graphlets.
In our experiments we set $\bar c = 1000$, which seems sufficient to give good accuracies on most graphlets.

\renewcommand{\thealgorithm}{}
\begin{algorithm}[h!]
\caption{AGS($\epsilon, \delta$)}
\begin{algorithmic}[1]
\small
\State $(c_1,\ldots,c_s) \leftarrow (0,\ldots,0)$ \Comment{graphlet counts}
\State $(w_1,\ldots,w_s) \leftarrow (0,\ldots,0)$ \Comment{graphlet weights}
\State $\cerr \leftarrow \lceil\frac{4}{\epsilon^2}\ln(\frac{2s}{\delta})\rceil$ \Comment{covering threshold}
\State $C \leftarrow \emptyset$ \Comment{graphlets covered}
\State{$T_j \leftarrow$ an arbitrary treelet type}
\While{$|C| < s$}
\For{each $i'$ in $1,\ldots,s$}
\State $w_{i'} \leftarrow w_{i'} + \sigma_{ji'}/r_j$ \label{ags:w_j}
\EndFor
\State $T_G \leftarrow$ an occurrence of $T_j$ drawn u.a.r.\ in $G$
\State $H_i \leftarrow$ the graphlet type spanned by $T_G$ \label{ags:h_j}
\State $c_i \leftarrow c_i + 1$ \label{ags:c_j}
\If{$c_i \ge \cerr$} \Comment{switch to a new treelet $T_j$}
  \State $C \leftarrow C \cup i$
  \State $j^* \leftarrow \arg \min_{j=1,\ldots,r} \frac{1}{r_j} \sum_{j \in C} \sigma_{ij} \, c_i / w_i$ \label{ags:estim}
  \State $T_j \leftarrow T_{j^*}$ 
\EndIf
\EndWhile
\State \textbf{return} $(\frac{c_1}{w_1}, \ldots, \frac{c_s}{w_s})$
\end{algorithmic}
\end{algorithm}

\subsection{Approximation guarantees}
We prove that, if AGS chooses the ``right'' treelet $T_{j^*}$, then we obtain multiplicative error guarantees.
Formally:
\begin{theorem}
\label{THM:AGS_APX}
If the tree $T_{j^*}$ chosen by AGS at line~\ref{ags:estim} minimizes $\prob[$\sample$(T_j)$ spans a copy of some $H_i \in C]$ then, with probability $(1-\delta)$, when AGS stops $c_i/w_i$ is a multiplicative $(1\pm\epsilon)$-approximation of $g_i$ for all $i=1,\ldots,s$.
\end{theorem}
The proof requires a martingale analysis and is deferred to Appendix~\ref{apx:proof_ags_apx}.
We stress that the guarantees hold for all graphlets, irrespective of their relative frequency.
In practice, AGS gives accurate counts for many or almost all graphlets at once, depending on the graph (see Section~\ref{sec:exp}).

\subsection{Sampling efficiency}
Let us turn to the sampling efficiency of AGS.
We start by showing that, on some graphs, AGS does no better than naive sampling, but that this holds for \emph{any} algorithm based on \sample($T$).
Formally:
\begin{theorem}
\label{thm:samplelb}
There are graphs $G$ where some graphlet $H$ represents a fraction $p_H = 1/\operatorname{poly}(n)$ of all graphlet copies, and any algorithm needs $\Omega(1/p_H)$ invocations of \sample$(T)$ in expectation to just find one copy of $H$.
\end{theorem}
\begin{proof}
Let $T$ and $H$ be the path on $k$ nodes.
Let $G$ be the $(n-k+2,k-2)$ lollipop graph; so $G$ is formed by a clique on $n-k+2$ nodes and a dangling path on $k-2$ nodes, connected by an arc.
$G$ contains $\Theta(n^k)$ non-induced occurrences of $T$ in $G$, but only $\Theta(n)$ induced occurrences of $H$ (all those formed by the $k-2$ nodes of the dangling path, the adjacent node of the clique, and any other node in the clique).
Since there are at most $\Theta(n^k)$ graphlets in $G$, then $H$ forms a fraction $p_H=\Theta(n^{1-k})$ of these.
Obviously $T$ is the only spanning tree of $H$; however, an invocation of $\operatorname{sample}(G,T)$ returns $H$ with probability $\Theta(n^{1-k})$ and thus we need $\Theta(n^{k-1})=\Theta(1/p_H)$ samples in expectation before obtaining $H$.
One can make $p_H$ larger by considering the $(n',n-n')$ lollipop graph for larger values of $n'$.
\end{proof}
We remark that Theorem~\ref{thm:samplelb} applies to the algorithms of~\cite{Jha&2015,Wang&2015,Wang&2016}, as they are are based on \sample($T$).

Since we cannot give good \emph{absolute} bounds on the samples, we analyse AGS against an optimal, clairvoyant adversary based on \sample($T$).
This adversary is clairvoyant in the sense that it knows exactly how many \sample($T_j$) calls to make for every treelet $T_j$ in order to get the desired bounds with the minimum number of calls.
Formally, we prove:
\begin{theorem}
\label{THM:AGS_COST}
If the tree $T_{j^*}$ chosen by AGS at line~\ref{ags:estim} minimizes $\prob[$\sample$(T_j)$ spans a copy of some $H_i \in C]$, then AGS makes a number of calls to \sample$()$ that is at most $O(\ln(s))=O(k^2)$ times the minimum needed to ensure that every graphlet $H_i$ appears in $\bar c$ samples in expectation.
\end{theorem}
The proof of the theorem relies on a fractional set cover and can be found in Appendix~\ref{apx:proof_ags_cost}.

\section{Experimental results}
\label{sec:exp}
In this section we compare the performance of \motivo\ to CC~\cite{Bressan&2018b} which, as said, is the current state of the art.
For readability, we give plots for a subset of datasets that are representative of the entire set of results.

\textbf{Set-up.}
We ran all our experiments on a commodity machine equipped with 64GB of main memory and 48 Intel Xeon E5-2650v4 cores at 2.5GHz with 30MB of L3 cache.
We allocated 880GB of secondary storage on a Samsung SSD850 solid-state drive, dedicated to the treelet count tables of \motivo.
Table~\ref{tab:graphs} shows the $9$ publicly available graphs on which we tested \motivo, and the largest tested value of $k$.
All graphs were made undirected and converted to the \motivo\ binary format.
For each graph we ran \motivo\ for all $k=5,6,7,8,9$, or until the build time did not exceed $1.5$ hours; except for \twitter\ and \lj, where we did $k=5,6$ regardless of time.
\begin{table}[h]
{\small
\begin{tabular}{lrrll}
graph & M nodes & M edges & {source} & {k} \\ 
\hline\\[-6pt]
\facebook & $0.1$ & $0.8$ & \SWS & 9\\
\bstan & $0.7$ & $6.6$ &  \SNAP & 9\\
\amazon & $0.7$ & $3.5$ & \SNAP & 9\\
\dblp & $0.9$ & $3.4$ & \SNAP & 9\\
\orkut & $3.1$ & $117.2$ & \SWS & 7\\
\lj & $5.4$ & $49.5$ & \LAW & 8\\
\yelp & $7.2$ & $26.1$ & \YLP & 8\\
\twitter & $41.7$ & $1202.5$ & \LAW & 6 (7$^*$)\\
\fs &  $65.6$ & $1806.1$ & \SNAP & 6 (7$^*$)
\end{tabular}
\caption{our graphs ($^*$ = with biased coloring)}
\label{tab:graphs}
}
\end{table}

\textbf{Ground truth.}
We computed exact 5-graphlet counts for \facebook, \dblp, \amazon, \lj\ and \orkut\ by running the ESCAPE algorithm~\cite{Pinar&2017}.
On the remaining graphs ESCAPE died by memory exhaustion or did not return within 24 hours.
For $k>5$ and/or larger graphs, we averaged the counts given by \motivo\ over $20$ runs, $10$ using naive sampling and $10$ using AGS.

\subsection{Computational performance}
\label{sub:perf}

\textbf{Build-up time.}
The table below shows the speedup of \motivo's build-up phase over CC's build-up phase (biased coloring is disabled).
Dashes mean CC failed by memory exhaustion or $64$-bit integer overflow (recall that \motivo\ works with 128-bit counters).
We removed \twitter\ and \fs\ since CC failed even for $k=5$.
Note that \motivo\ is $2\times$-$5\times$ faster than CC on $5$ out of $7$ graphs, and never slower on the other ones.
\begin{table}[h!]
\centering
\small
\begin{tabular}{rrrrrr}
graph & \multicolumn{1}{c}{k=$5$} & \multicolumn{1}{c}{k=$6$} & \multicolumn{1}{c}{k=$7$} & \multicolumn{1}{c}{k=$8$} & \multicolumn{1}{c}{k=$9$} \\
\midrule
\facebook\ & 2.9 &	2.9 & 3.3 & 4.8 & 3.5 \\
\bstan\ & 2.0 & - & - & - & - \\
\amazon\ & 1.6 & 1.3 & 1.2 & 1.0 & 1.1 \\
\dblp\ & 1.5 & 1.2 & 1.3 & 1.4 & 1.8  \\
\orkut\ & 4.5 & - & - &  & \\
\lj\ & 2.8 & 3.1 & - & - &\\
\yelp & 2.4 & - & - & - & 
\end{tabular}
\end{table}

\textbf{Count table size.}
The table below shows the ratio between the main memory footprint of CC and the total external memory usage of \motivo; both are indicators of the total count table size.
The footprint of CC is computed as the smallest JVM heap size that allowed it to run.
In almost all cases \motivo\ saves a factor of $2$, in half of the cases a factor of $5$, and on \yelp, the largest graph managed by CC, a factor of $8$.
For $k=7$, CC failed on 6 over 9 graphs, while \motivo\ processed all of them with a space footprint of less than 12GB (see below).

\begin{table}[h!]
\centering
\small
\begin{tabular}{rrrrrr}
graph & \multicolumn{1}{c}{k=$5$} & \multicolumn{1}{c}{k=$6$} & \multicolumn{1}{c}{k=$7$} & \multicolumn{1}{c}{k=$8$} & \multicolumn{1}{c}{k=$9$} \\
\midrule
\facebook\ & 108.7 & 87.5 & 54.5 & 17.3 & 7.1 \\
\bstan\ & 36.5 & - & - & - & - \\
\amazon\ & 6.6 & 3.3 & 2.1 & 1.1 & 1.0 \\
\dblp\ & 6.2 & 4.0 & 2.4 & 1.1 & 1.0 \\
\orkut\ & 8.5 & - & - & & \\
\lj\ & 11.4 & 3.8 & - & - & \\
\yelp\ & 8.0 & - & - & - & 
\end{tabular}
\end{table}

\textbf{Sampling speed.}
Finally, we compare the sampling speed of \motivo\ (without AGS) versus CC.
\motivo\ is always $10\times$ faster, and even $100\times$ faster on \yelp, the largest graph managed by CC, and the gap often diverges with $k$.
Obviously, this means \motivo\ gives more accurate estimates for a fixed time budget.
Note also that \motivo\ is faster even though it has to access the tables on disk.
\begin{table}[h!]
\centering
\small
\begin{tabular}{rrrrrr}
graph & \multicolumn{1}{c}{k=$5$} & \multicolumn{1}{c}{k=$6$} & \multicolumn{1}{c}{k=$7$} & \multicolumn{1}{c}{k=$8$} & \multicolumn{1}{c}{k=$9$} \\
\midrule
\facebook\ & 14.9 & 13.2 & 9.4 & 50.0 & 115.9 \\
\bstan\ & 29.3 & - & - & - & - \\
\amazon\ & 60.7 & 12.6 & 13.2 & 13.2 & 16.5 \\
\dblp\ & 17.7 & 11.4 & 10.1 & 44.8 & 88.6 \\
\orkut\ & 29.2 & - & - & & \\
\lj\ & 31.8 & 28.5 & - & - & \\
\yelp\ & 159.6 & - & - & - & 
\end{tabular}
\end{table}

\textbf{Final remarks.}
\motivo\ runs in minutes on graphs that CC could not even process (\bstan, \orkut, \yelp\ for $k=6$), and runs in less than one hour for all but the largest instance.
To put it in perspective, ESCAPE~\cite{Pinar&2017} and the random-walk algorithms of~\cite{Bhuiyan&2012,Chen&2016,Han&2016,Wang&2014} can take entire days even for graphs $10-100$ times smaller. We also note that, very differently from all these algorithms, \motivo\ is predictable as a function of $m$ and $k$.
One can see this by looking at the running time per million edge, and at the space usage per node, shown in Figure~\ref{fig:perf}.

\begin{figure}[h]
\centering
\includegraphics[scale=.66]{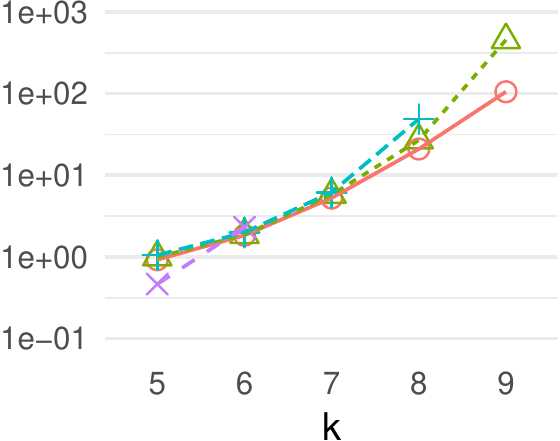}
\hfill
\includegraphics[scale=.66]{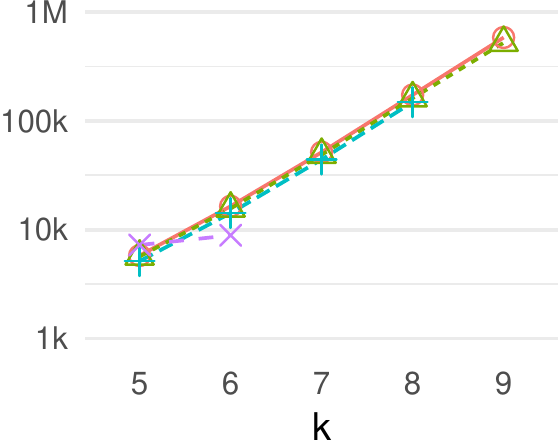}
\\[8pt]
\includegraphics[scale=.66]{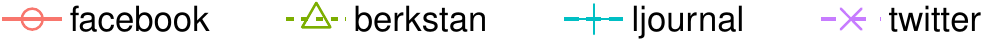}
\caption{Motivo's build-up time (seconds per million edge) and space usage (bits per input node).
}
\label{fig:perf}
\end{figure}

\subsection{Accuracy}
\label{sub:accuracy}
The previous section showed how \motivo\ scales to instances much larger than the state of the art.
We now show \motivo\ is also far more accurate in its estimates, in particular using AGS.
For a comparison against CC, note that in fact CC itself is strictly dominated by the ``naive sampling'' algorithm of \motivo.
Indeed, the standard sampling strategy of \motivo\ is \emph{exactly} the one of CC described in Section~\ref{sub:sampling}.
However, as shown above \motivo's implementation is much faster, hence takes many more samples and is consequently more accurate.
Therefore here we report the accuracy of the naive sampling of \motivo\ (which can be seen as an efficient version of CC), and compare it against AGS.
All plots below report the average over $10$ runs, with whiskers for the 10\% and 90\% percentiles.
Naive sampling is shown by the left bars, and AGS by the right bars.

A final remark.
The accuracy of estimates obviously depends on the number of samples taken.
One option would be to fix an absolute budget, say 1M samples.
Since however for $k=5$ there are only 21 distinct graphlets and for $k=8$ there are over 10k, we would certainly have much higher accuracy in the first case.
As a compromise, we tell \motivo\ to spend in sampling the same amount of time taken by the build-up phase.
This is also what an ``optimal'' time allocation strategy would do -- statistically speaking, if we have a budget of 100 seconds and the build-up takes 5 seconds, we would perform 10 runs with 5 seconds of sampling each and average over the estimates.

\textbf{Error in $\ell_1$ norm.}
First, we evaluate how accurately  \motivo\ reconstructs the global $k$-graphlet distribution.
If $\mathbf{f} = (f_1,\ldots,f_s)$ are the ground-truth graphlet frequencies, and $\mathbf{\hat f} = (\hat f_1,\ldots,\hat f_s)$ their estimates, then the $\ell_1$ error is $\ell_1(\mathbf{f}, \mathbf{\hat f}) = \sum_{i=1}^s |\hat f_i - f_i|$.
In our experiments, the $\ell_1$ error was below $5\%$ in all cases, and below $2.5\%$ for all $k \le 7$.

\newcommand{\err}{\operatorname{err}}
\textbf{Single-graphlet count error.}
The count error of $H$ is:
\begin{align}
\err_H = \frac{\hat{c}_H - c_H}{c_H}
\end{align}
where $c_H$ is the ground-truth count of $H$ and $\hat{c}_H$ its estimate.
Thus $\err_H=0$ means a perfect estimate, and $\err_H=-1$ means the graphlet is missed.
Figure~\ref{fig:errdist} shows the distribution of $\err_H$ for one run, for naive sampling (top) and AGS (bottom), as $k$ increases from left to right.
We see that (1) AGS is much more accurate, especially for larger values of $k$, and (2) CC's naive sampling misses many graphlets.
(The advantage of AGS over naive sampling is discussed in more detail below).
Inevitably, the error spreads out with $k$; recall that the total number of distinct $8$-graphlets is over $10^4$.
\newcommand{\myscale}{.8}
\begin{figure*}[h!]
\centering
\includegraphics[scale=\myscale]{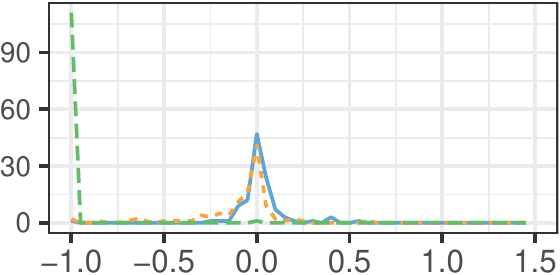}\hspace*{15pt}
\includegraphics[scale=\myscale]{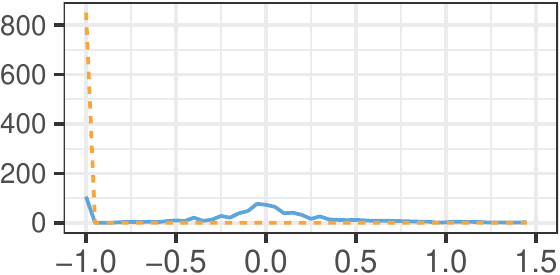}\hspace*{15pt}
\includegraphics[scale=\myscale]{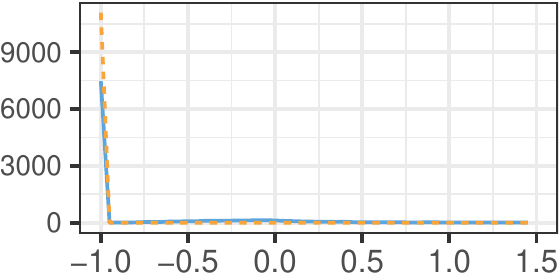}
\\[10pt]
\includegraphics[scale=\myscale]{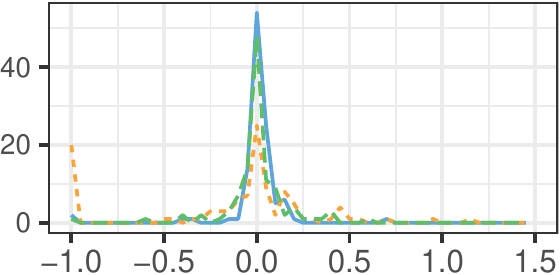}\hspace*{15pt}
\includegraphics[scale=\myscale]{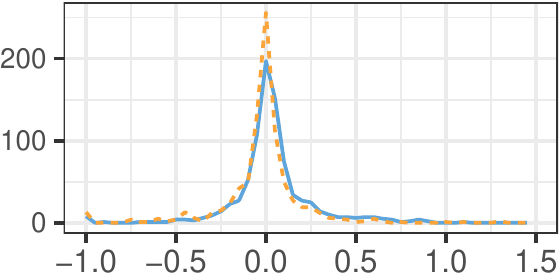}\hspace*{15pt}
\includegraphics[scale=\myscale]{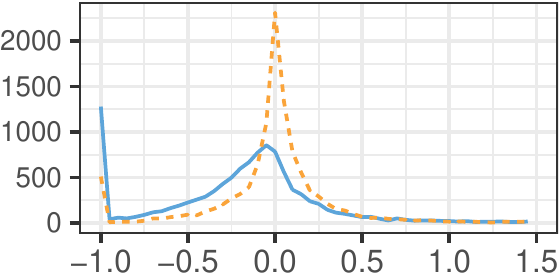}
\\[10pt]
\includegraphics[scale=0.7]{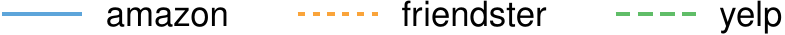}
\caption{distribution of graphlet count error for $k=6,7,8$.
Top: naive sampling.
Bottom: AGS.}
\label{fig:errdist}
\end{figure*}

\textbf{Number of accurate graphlets.}
For a complementary view, we consider the number of graphlets whose estimate is within $\pm $50\% of the ground-truth value (Figure~\ref{fig:n_ok}).
This number easily often reaches the thousands, and for $k=9$ even hundreds of thousands (note that the plot is in log-scale).
We remind the reader that all this is carried out in minutes or, in the worst case, in two hours.
Alternatively, we can look at these numbers in relative terms, that is, as a fraction of the total number of distinct graphlets in the ground truth (Figure~\ref{fig:f_ok}).
On all graphs except \bstan, this ratio is over 90\% of graphlets for $k=6$, over 75\% of graphlets for $k=7$, and over 50\% of graphlets for $k=8$, for either naive sampling or AGS.
The choice of $50\%$ has the sole purpose of deliver the picture; however, note that such an error is achieved for thousand of graphlets at once, which moreover have counts differing by many orders of magnitude.
\begin{figure*}[h!]
\centering
\includegraphics[scale=.75]{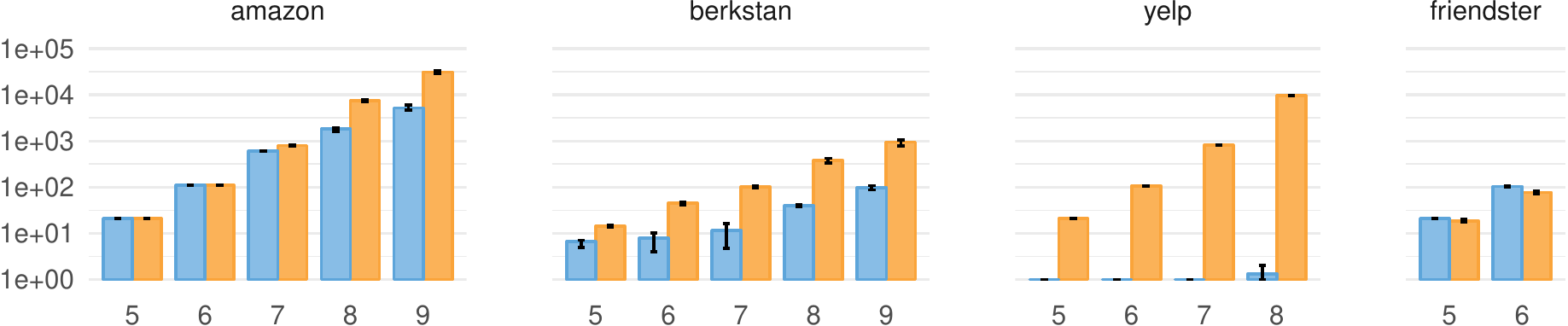}
\vspace*{3pt}
\includegraphics[scale=.74]{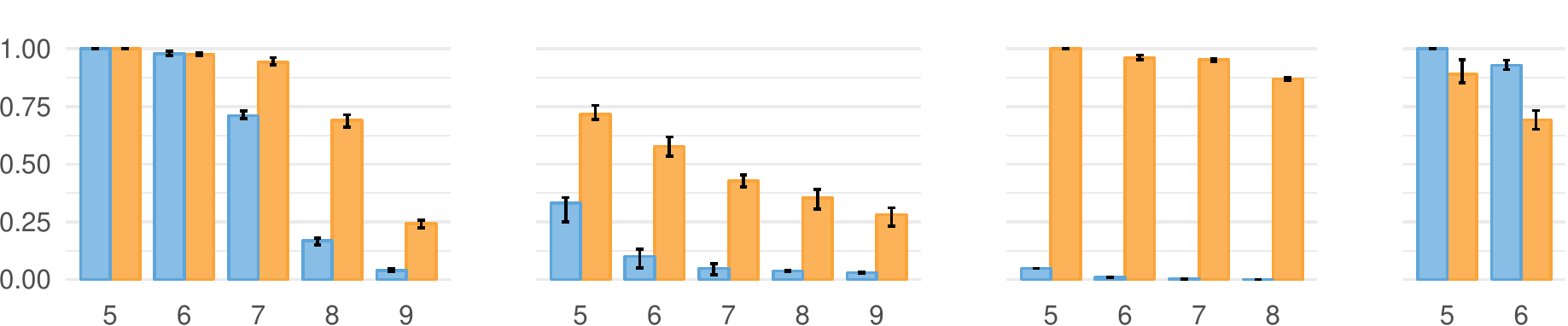}
\caption{graphlet counts with error within $\pm 50\%$.
Top: absolute number, in logarithmic scale.
Bottom: as a fraction of the total.}
\label{fig:f_ok}
\label{fig:n_ok}
\end{figure*}

\subsection{Performance of AGS}
Finally, we show how  AGS outperforms naive sampling, as anticipated.
The best example is the \yelp\ graph.
For $k=8$, over $99.9996\%$ of the $k$-graphlets are stars; and, as one can expect, naive sampling finds only the star, and thus gives accurate estimates for only $1$ graphlet, or 0.01\% of the total -- see Figure~\ref{fig:n_ok}. 
In other terms, naive sampling misses 9999 graphlets out of 10000.
However, AGS returns estimates within $50\%$ error for 9645 graphlets, or 87\% of the total.
Note also that the sampling rate of AGS is approximately 20 times higher than that of naive sampling for this dataset.
A complementary perspective is given in Figure~\ref{fig:minfreq}, which shows the frequency of the rarest graphlet that appeared in at least $10$ samples (to filter out those appearing just by chance).
For \yelp, the only graphlet found by naive sampling has frequency $99.9996\%$ (the star) while AGS always finds graphlets with frequency below $10^{-21}$.
To give an idea, imagine that for those graphlets naive sampling would need $\approx 3\cdot 10^3$ years even by taking $10^9$ samples per second.

Let us make a final remark.
On some graphs, AGS is slightly worse than naive sampling.
This is expected: AGS is designed for skewed graphlet distributions, and loses ground on flatter ones.
As a sanity check, we computed the $\ell_2$ norm of the graphlet distributions.
The three graphs where AGS beats naive sampling by a largest margin, \bstan, \yelp\ and \twitter, have for all $k$ the highest $\ell_2$ norms ($ > .99$).
Symmetrically, \facebook, \dblp\ and \fs, have for all $k$ the three lowest $\ell_2$ norms, and there AGS performs slightly worse than naive sampling.

\begin{figure*}[h]
\centering
\includegraphics[scale=.75]{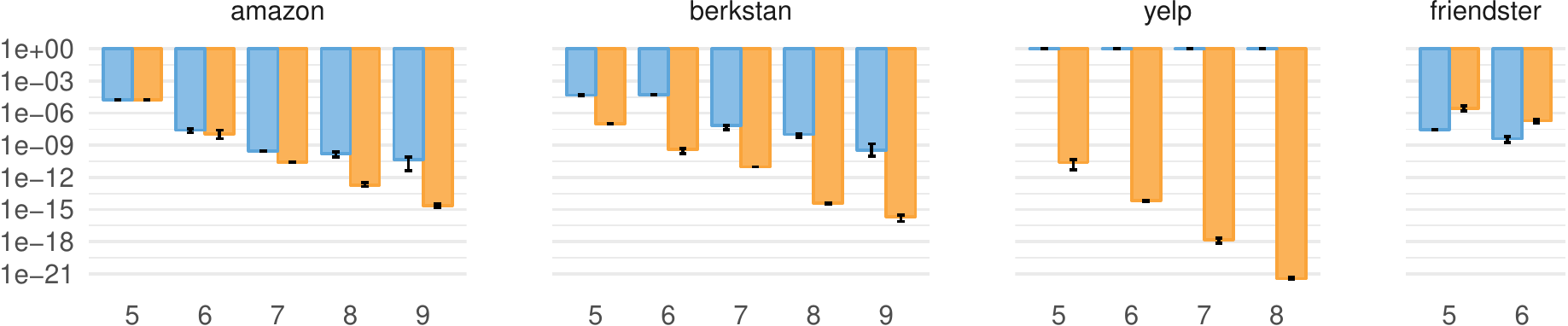}
\caption{frequency of the rarest graphlet appearing in 10 or more samples.}
\label{fig:minfreq}
\end{figure*}

\section{Conclusions}
Color coding is a versatile technique that can be harnessed to scale motif counting to truly massive graphs, with tens of millions of nodes and billions of edges, and with approximation guarantees previously out of reach.
Although we have made steps towards an efficient use of color coding, further optimizations are certainly possible.
It would especially interesting to investigate principled ways to reduce space usage, which is still a bottleneck of this approach.

\appendix

\section{Proof of Theorem~\ref{THM:CONC_DEP}}
\label{apx:conc_dep}
We use a concentration bound for dependent random variables from~\cite{Dubhashi&2009}.
Let $\vk_i$ be the set of copies of $H_i$ in $G$.
For any $h \in \vk_i$ let $X_{h}$  be the indicator random variable of the event that $h$ becomes colorful.
Let $c_i = \sum_{h \in \vk_i} X_h$; clearly $\E[c_i] = p_k |\vk_i| = p_k n_i$.
Note that for any $h_1,h_2 \in \vk_i$, $X_{h_1}, X_{h_2}$ are independent if and only if $|V(h_1) \cap V(h_2)| \le 1$ i.e.\ if $h_1,h_2$ share at most one node.
For any $u,v \in G$ let then $g(u,v) = |\{h \in \vk_i : u,v \in h \}|$, and define $\chi_k = 1 + \max_{u,v \in G} g(u,v)$.
By standard counting argument one can see that $\max_{u,v \in G} g(u,v) \le (k-1)!\Delta^{k-2}-1$ and thus $\chi_k \le (k-1)!\Delta^{k-2}$.
The bound then follows immediately from Theorem~3.2 of~\cite{Dubhashi&2009} by setting $t= \epsilon c_i$, $(b_{\alpha} - a_{\alpha}) = 1$ for all $\alpha = h \in \vk_i$, and $\chi^*(\Gamma) \le \chi_k \le (k-1)!\Delta^{k-2}$.

\section{Proof of Theorem~\ref{THM:AGS_APX}}
\label{apx:proof_ags_apx}
The proof requires a martingale analysis, since the distribution from which we draw the graphlets changes over time.
We make use of a martingale tail inequality originally from~\cite{Freedman1975} and stated (and proved) in the following form in~\cite{Alon&2010}, page 1476:
\begin{theorem}[\cite{Alon&2010}, Theorem 2.2]
\label{thm:alon}
Let $(Z_0,Z_1,\ldots)$ be a martingale with respect to the filter $(\mathcal{F}_\tau)$. Suppose that $Z_{\tau+1}-Z_\tau \le M$ for all $\tau$, and write $V_t = \sum_{\tau=1}^t \var{Z_\tau|\mathcal{F}_{\tau-1}}$. Then for any $z,v>0$ we have:
\begin{align}
\prob\!\big[\exists \, t  : Z_t \ge Z_0 + z, V_t \le v\big] \le \exp\!{\Big[\!-\frac{z^2}{2(v+Mz)}\Big]}
\end{align}
\end{theorem}
We now plug into the formula of Theorem~\ref{thm:alon} the appropriate quantities.
In what follows we \emph{fix} a graphlet $H_i$ and analyse the concentration of it estimate.
Unless necessary, we drop the index $i$ from the notation.
\\
\textbf{A.} For $t \ge 1$ let $X_t$ be the indicator random variable of the event that $H_i$ is the graphlet sampled at step $t$ (line~\ref{ags:h_j} of AGS).\\
\textbf{B.} For $t \ge 0$ let $Y_j^t$ be the indicator random variable of the event, at the end of step $t$, the treelet to be sampled at the next step is $T_j$.\\
\textbf{C.} For $t \ge 0$ let $\mathcal{F}_t$ be the event space generated by the random variables $Y_j^{\tau} : j\in[r], \, \tau=0,\ldots,t$.
For any random variable $Z$, then, $\E[Z \,| \,\mathcal{F}_t] = \E[Z \, | \, Y_j^{\tau} : j\in[r], \, \tau =0,\ldots,t]$, and $\var{Z\,|\,\mathcal{F}_t}$ is defined analogously. \\
\textbf{D.} For $t \ge 1$ let $P_t = \E[X_t|\mathcal{F}_{t-1}]$ be the probability that the graphlet sampled at the $t$-th invocation of line~\ref{ags:h_j} is $H_i$, as a function of the events up to time $t-1$.
It is immediate to see that $P_t = \sum_{j=1}^r Y_j^{t-1} a_{ji}$.
\\
\textbf{E.} Let $Z_0 = 0$, and for $t \ge 1$ let $Z_t = \sum_{\tau=1}^t (X_t - P_t)$.
Now, $(Z_t)_{t \ge 0}$ is a martingale with respect to the filter $(\mathcal{F}_t)_{t\ge 0}$, since $Z_t$ is obtained from $Z_{t-1}$ by adding $X_t$ and subtracting $P_t$ which is precisely the expectation of $X_t$ w.r.t.\ $\mathcal{F}_{t-1}$.
\textbf{F.} Let $M=1$, since $|Z_{t+1}-Z_t| = |X_{t+1} - P_t| \le 1$ for all $t$.\\

Finally, notice that $\var{Z_t|\mathcal{F}_{t-1}} = \var{X_t|\mathcal{F}_{t-1}}$, since again $Z_t = Z_{t-1} + X_t - P_t$, and both $Z_{t-1}$ and $P_t$ are a function of $\mathcal{F}_{t-1}$, so their variance w.r.t.\ $\mathcal{F}_{t-1}$ is $0$.
Now, $\var{X_t|\mathcal{F}_{t-1}} = P_t(1-P_t) \le P_t$; and therefore we have $V_t = \sum_{\tau=1}^t \var{Z_\tau \,|\,\mathcal{F}_{\tau-1}} \le \sum_{\tau=1}^t P_{\tau}$.
Then by Theorem~\ref{thm:alon}:
\begin{lemma}
\label{cor:martbound}
For all $z,v > 0$ we have
\begin{align}
\prob\!\big[\exists \, t : Z_t \ge z, \sum_{\tau=1}^t P_\tau \le v \big] &\le 
\exp{\!\Big[\!-\!\frac{z^2}{2(v+z)}\Big]}
\end{align}
\end{lemma}
Consider now AGS($\epsilon, \delta$).
Recall that we are looking at a \emph{fixed} graphlet $H_i$ (which here does \emph{not} denote the graphlet sampled at line~\ref{ags:h_j}).
Note that $\sum_{\tau=1}^t X_{\tau}$ is exactly the value of $c_i$ after $t$ executions of the main cycle (see line~\ref{ags:c_j}).
Similarly, note that $\sum_{\tau=1}^t P_\tau$ is the value of $g_i \cdot w_i$ after $t$ executions of the main cycle: indeed, if $Y_j^{t-1}=1$, then at step $\tau$ we add to $w_i$ the value $\frac{\sigma_{ij}}{r_j}$ (line~\ref{ags:w_j}), while the probability that a sample of $T_j$ yields $H_i$ is exactly $\frac{g_i\sigma_{ij}}{r_j}$.
Therefore, after the main cycle has been executed $t$ times, $Z_t = \sum_{\tau=1}^t (X_t - P_t)$ is the value of $c_i - g_i w_i$.

Now to the bounds.
Suppose that, when AGS($\epsilon, \delta$) returns, $\frac{c_i}{w_i} \ge g_i(1 + \epsilon)$ i.e.\ $c_i(1 - \frac{\epsilon}{1+\epsilon}) \ge g_i w_i$.
On the one hand this implies that $c_i - g_iw_i \ge c_i\frac{\epsilon}{1+\epsilon}$ i.e.\ $Z_t \ge c_i\frac{\epsilon}{1+\epsilon}$; and since upon termination $c_i = \cerr$, this means $Z_t \ge \cerr\frac{\epsilon}{1+\epsilon}$.
On the other hand it implies $g_i w_i \le c_i(1 - \frac{\epsilon}{1+\epsilon})$ i.e.\ $\sum_{\tau=1}^t P_\tau \le c_i(1 - \frac{\epsilon}{1+\epsilon})$; again since upon termination $c_i = \cerr$, this means $\sum_{\tau=1}^t P_\tau \le \cerr(1 - \frac{\epsilon}{1+\epsilon})$.
We can then invoke Lemma~\ref{cor:martbound} with $z=\cerr\frac{\epsilon}{1+\epsilon}$ and $v=\cerr(1 - \frac{\epsilon}{1+\epsilon})$, and since $v+z=\cerr$ we get:
\begin{align}
\prob\!\Big[\frac{c_i}{w_i} \ge g_i(1+\epsilon)\Big]
&\le \exp{\!\Big[\!-\!\frac{(\cerr\frac{\epsilon}{1+\epsilon})^2}{2\cerr}\Big]}\\
&= \exp{\!\Big[\!-\!\frac{\epsilon^2 \cerr}{2(1+\epsilon)^2}\Big]}
\end{align}
but $\frac{\epsilon^2 \cerr}{2(1+\epsilon)^2} \ge \frac{\epsilon^2}{2(1+\epsilon)^2} \frac{4}{\epsilon^2}\ln\!\big(\frac{2s}{\delta}\big) \ge \ln\!\big(\frac{2s}{\delta}\big)$ and thus the probability above is bounded by $\frac{\delta}{2s}$.

Suppose instead that, when AGS($\epsilon, \delta$) returns, $\frac{c_i}{w_i} \le g_i(1 - \epsilon)$ i.e.\ $c_i(1 + \frac{\epsilon}{1-\epsilon}) \le g_i w_i$.
On the one hand this implies that $c_i - g_i w_i \ge \frac{\epsilon}{1-\epsilon} c_i$, that is, upon termination we have $-Z_t \ge \frac{\epsilon}{1-\epsilon}\cerr$.
Obviously $(-Z_t)_{t\ge 0}$ is a martingale too with respect to the filter $(\mathcal{F}_t)_{t\ge 0}$, and therefore Lemma~\ref{cor:martbound} still holds if we replace $Z_t$ with $-Z_t$.
Let then $t_0 \le t$ be the first step where $-Z_{t_0} \ge \frac{\epsilon}{1-\epsilon} \cerr$; since $|Z_t - Z_{t-1}| \le 1$, it must be $-Z_{t_0} < \frac{\epsilon}{1-\epsilon}\cerr + 1$.
Moreover $\sum_{\tau=1}^t X_\tau$ is nondecreasing in $t$, so $\sum_{\tau=1}^{t_0} X_\tau \le \cerr$.
It follows that $\sum_{\tau=1}^{t_0} P_\tau = -Z_{t_0} + \sum_{\tau=1}^{t_0} X_\tau < \frac{\epsilon}{1-\epsilon} \cerr + 1 + \cerr = \frac{1}{1-	\epsilon} \cerr +1$.
Invoking again Lemma~\ref{cor:martbound} with $z=\frac{\epsilon}{1-\epsilon} \cerr$ and $v=\frac{1}{1-\epsilon} \cerr +1$, we obtain:
\begin{align}
\prob\!\big[\frac{c_i}{w_i} \le g_i(1-\epsilon)\big]
&\le \exp{\!\Big[\!-\!\frac{(\cerr\frac{\epsilon}{1-\epsilon})^2}{2(\frac{1+\epsilon}{1-\epsilon}\cerr+1)}\Big]}\\
&\le \exp{\!\Big[\!-\!\frac{\epsilon^2 \cerr^2}{2(1+\cerr)}\Big]}
\end{align}
but since $\cerr \ge 4$ then $\frac{\cerr}{1 + \cerr} \ge \frac{4}{5}$ and so $\frac{\epsilon^2 \cerr^2}{2(1+\cerr)} \ge \frac{2 \epsilon^2 \cerr}{5}$.
By replacing $\cerr$ we get $\frac{2 \epsilon^2 \cerr}{5} \ge \frac{2 \epsilon^2}{5} \frac{4}{\epsilon^2}\ln\!\big(\frac{2s}{\delta}\big) > \ln\!\big(\frac{2s}{\delta}\big)$ and thus once again the probability of deviation is bounded by $\frac{\delta}{2s}$.

By a simple union bound, the probability that $\frac{c_i}{w_i}$ is not within a factor $(1 \pm \epsilon)$ of $g_i$ is at most $\frac{\delta}{s}$.
The theorem follows by a union bound on all $i \in [s]$.

\section{Proof of Theorem~\ref{THM:AGS_COST}}
\label{apx:proof_ags_cost}
For each $i \in [s]$ and each $j \in [r]$ let $a_{ji}$ be the probability that \sample($T_j$) returns a copy of $H_i$.
Note that $a_{ji} = g_i\sigma_{ij}/r_j$, the fraction of colorful copies of $T_j$ that span a copy of $H_i$.
Our goal is to allocate, for each $T_j$, the number $x_j$ of calls to \sample($T_j$), so that (1) the total number of calls $\sum_{j} x_j$ is minimised and (2) each $H_i$ appears at least $\cerr$ times in expectation.
Formally, let $\mathbf{A} = (a_{ji})^{\intercal}$, so that columns correspond to treelets $T_j$ and rows to graphlets $H_i$, and let $\bfx = (x_1,\ldots,x_r) \in \mathbb{N}^r$.
We obtain the following integer program:
\begin{align*}
\left\{
\begin{array}{l}
\min \mathbf{1}^{\intercal} \bfx\\
\text{s.t.} \, \mathbf{A} \bfx \ge \cerr\,\mathbf{1}\\
\phantom{\text{s.t.}} \, \bfx \in \mathbb{N}^r
\end{array}
\right.
\end{align*}

We now describe the natural greedy algorithm for this problem; it turns out that this is precisely AGS.
The algorithm proceeds in discrete time steps.
Let $\bfx^0 = \mathbf{0}$, and for all $t \ge 1$ denote by $\bfx^t$ the partial solution after $t$ steps.
The vector $\mathbf{A} \bfx^t$ is an $s$-entry column whose $i$-th entry is the expected number of occurrences of $H_i$ drawn using the sample allocation given by $\bfx^t$.
We define the vector of residuals at time $t$ as $\bfc^t = \max(\mathbf{0}, \bfc - \mathbf{A} \bfx^t)$, and for compactness we let $c^t =  \mathbf{1}^{\intercal} \bfc^t$.
Note that $\bfc^0 = \cerr\,\mathbf{1}$ and $c^0 = s \cerr$.
Finally, we let $U^t = \{i : c^t_i > 0\}$; this is the set of graphlets not yet covered at time $t$, and clearly $U^0 = [s]$.


At the $t$-th step the algorithm chooses the $T_{j^*}$ such that \sample($T_{j^*}$) spans an uncovered graphlet with the highest probability, by computing:
\begin{align}
\label{eqn:istar}
j^* := \arg \max_{j=1,\ldots,r} \sum_{i \in U_t} a_{ji}
\end{align}
It then lets $\bfx^{t+1} = \bfx + \mathbf{e}_{j^*}$, where $\mathbf{e}_{j^*}$ is the indicator vector of $j^*$, and updates $\bfc^{t+1}$ accordingly.
The algorithm stops when $U^t = \emptyset$, since then $\bfx^t$ is a feasible solution.
We prove:
\begin{lemma}
\label{lem:greedy_cost}
Let $z$ be the cost of the optimal solution.
Then the greedy algorithm returns a solution of cost $O(z\ln(s))$.
\end{lemma}
\begin{proof}
Let $w_j^t = \sum_{i \in U_t} a_{ji}$ (note that this is a \emph{treelet} weight).
For any $j \in [r]$ denote by $\Delta_j^t = c^{t} - c^{t+1}$ the decrease in overall residual weight we would obtain if $j^* = j$.
Note that $\Delta_j^t \le w_j^t$.
We consider two cases.\\
\textbf{Case 1}: $\Delta_{j^*}^t < w_{j^*}^t$.
This means for some $i \in U_t$ we have $c_i^{t+1} = 0$, implying $i \notin U_{t+1}$.
In other terms, $H_i$ becomes covered at time ${t+1}$.
Since the algorithm stops when $U_t = \emptyset$, this case occurs at most $|U^0| = s$ times.
\\
\textbf{Case 2}: $\Delta_{j^*}^t = w_{j^*}^t$.
Suppose then the original problem admits a solution with cost $z$.
Obviously, the ``residual'' problem where $\mathbf{c}$ is replaced by $\mathbf{c}^{t}$ admits a solution of cost $z$, too.
This implies the existence of $j \in [r]$ with $\Delta_j^t \ge \frac{1}{z} c^t$, for otherwise any solution for the residual problem would have cost $> z$.
But by the choice of $j^*$ it holds $\Delta_{j^*} = w_{j^*}^t \ge w_{j}^t \ge \Delta_j^t$ for any $j$, hence $\Delta_{j^*}^t \ge \frac{1}{z} c^t$.
Thus by choosing $j^*$ we get $c^{t+1} \le (1-\frac{1}{z})c^t$.
After running into this case $\ell$ times, the residual cost is then at most $c^0 (1-\frac{1}{z})^\ell$.

Note that $\ell + s \ge c^0 = s \cdot \cerr$ since at any step the overall residual weight can decrease by at most $1$.
Therefore the algorithm performs $\ell + s = O(\ell)$ steps overall.
Furthermore, after $\ell + s$ steps we have $c^{\ell+s} \le s \cerr e^{-\frac{\ell}{z}}$, and by picking $\ell=z\ln(2s)$ we obtain $c^{\ell+s} \le \frac{\cerr}{s}$, and therefore each one of the $s$ graphlets receives weight at least $\frac{\cerr}{2}$.
Now, if we replace $\cerr\,\mathbf{1}$ with $ 2 \cerr\,\mathbf{1}$ in the original problem, the cost of the optimal solution is at most $2 z$, and in $O(z \ln(s))$ steps the algorithm finds a cover where each graphlet has weight at least $\cerr$.
\end{proof}
Now, note that the treelet index $j^*$ given by Equation~\ref{eqn:istar} remains unchanged as long as $U_t$ remains unchanged.
Therefore we need to recompute $j^*$ only when some new graphlet exits $U_t$ i.e.\ becomes covered.
In addition, we do not need each value $a_{ji}$, but only their sum $\sum_{i \in U_t} a_{ji}$.
This is precisely the quantity that AGS estimates at line~\ref{ags:estim}.
Theorem~\ref{THM:AGS_COST} follows immediately as a corollary.

\balance

\clearpage
\bibliographystyle{abbrv}
\bibliography{biblio}

\end{document}